\documentclass[11pt]{article}
\usepackage{fullpage}

\usepackage{times}
\usepackage{comment,amsfonts,amssymb,amsmath,amsthm,graphicx,algorithm,algorithmic}
\newcommand{\commentout}[1]{}
\usepackage{multirow}
\usepackage{boldline}

\ifx\pdftexversion\undefined
\usepackage[colorlinks,linkcolor=black,filecolor=black,citecolor=black,urlco
lor=black,pdfstartview=FitH]{hyperref}
\else
\usepackage[colorlinks,linkcolor=blue,filecolor=blue,citecolor=blue,urlcolor
=blue,pdfstartview=FitH]{hyperref}
\fi

\newcommand{\alert}[1]{\textbf{\color{red}
[[[#1]]]}\marginpar{\textbf{\color{red}**}}\typeout{ALERT:
\the\inputlineno: #1}}

\def\eps{\epsilon}
\def\tO{\tilde{O}}
\def\cP{{\cal P}}

\newcommand{\polylog}{{\rm polylog}}

\newcommand{\E}{{\mathbb{E}}}

\newcommand{\mommit}[1]{}
\newcommand{\namedref}[2]{\hyperref[#2]{#1~\ref*{#2}}}
\newcommand{\sectionref}[1]{\namedref{Section}{#1}}

\newcommand{\theoremref}[1]{\namedref{Theorem}{#1}}

\newcommand{\claimref}[1]{\namedref{Claim}{#1}}
\newcommand{\lemmaref}[1]{\namedref{Lemma}{#1}}

\newtheorem{theorem}{Theorem}
\newtheorem{lemma}{Lemma}
\newtheorem{corollary}[lemma]{Corollary}
\newtheorem{remark}{Remark}
\newtheorem{claim}[lemma]{Claim}

\usepackage{pdfsync}
\usepackage{authblk}
\begin{document}

\title{Linear-Size Hopsets with Small Hopbound, \\ and  Distributed Routing with Low Memory}

\author[1]{Michael Elkin\thanks{This research was supported by the ISF grant No. (724/15).}}
\author[1]{Ofer Neiman\thanks{Supported in part by ISF grant No. (523/12) and by BSF grant No. 2015813.}}

\affil[1]{Department of Computer Science, Ben-Gurion University of the Negev,
Beer-Sheva, Israel. Email: \texttt{\{elkinm,neimano\}@cs.bgu.ac.il}}

\date{}
\maketitle

\begin{abstract}
For a positive parameter $\beta$,  the {\em $\beta$-bounded distance} between a pair of vertices $u,v$ in a weighted undirected graph $G = (V,E,\omega)$ is the length of the shortest $u-v$ path in $G$ with at most $\beta$ edges, aka {\em  hops}.
For $\beta$ as above and $\eps > 0$, a {\em $(\beta,\eps)$-hopset} of $G = (V,E,\omega)$ is a graph $G' =(V,H,\omega_H)$ on the same vertex set, such that all distances in $G$ are $(1+\eps)$-approximated by $\beta$-bounded distances in $G \cup G'$.

Hopsets are a fundamental graph-theoretic and graph-algorithmic construct, and they are widely used for distance-related problems in a variety of computational settings. Currently existing constructions of hopsets  produce hopsets either with $\Omega(n \log n)$ edges, or with a hopbound $n^{\Omega(1)}$. In this paper we devise a construction of {\em linear-size} hopsets with hopbound  (ignoring the dependence on $\eps$) $(\log n)^{\log^{(3)} n + O(1)}$.  This improves the previous bound almost exponentially.

We also devise efficient implementations of our construction in PRAM and distributed settings. The only existing PRAM algorithm \cite{EN16} for computing hopsets with a constant (i.e., independent of $n$) hopbound requires $n^{\Omega(1)}$ time. We devise a PRAM algorithm with polylogarithmic running time for computing hopsets with a constant hopbound, i.e., our running time is {\em exponentially} better than the previous one.  Moreover, these hopsets are also significantly sparser than their counterparts from
 \cite{EN16}.

We use our hopsets to devise a distributed routing scheme that exhibits near-optimal tradeoff between individual memory requirement $\tO(n^{1/k})$ of vertices throughout preprocessing and routing phases of the  algorithm, and stretch $O(k)$, along with a near-optimal construction time $\approx D + n^{1/2 + 1/k}$, where $D$ is the hop-diameter of the input graph. Previous distributed routing algorithms either suffered  from a prohibitively large memory requirement $\Omega(\sqrt{n})$, or had a near-linear construction time, even on graphs with small hop-diameter $D$.
\end{abstract}

\thispagestyle{empty}
\newpage
\setcounter{page}{1}
\section{Introduction}
\label{sec:introd}

\subsection{Hopsets}

Consider a weighted undirected graph $G = (V,E,\omega)$. Consider another graph $G_H = (V,H,\omega_H)$ on the same vertex set, that satisfies that for every $(u,v) \in H$, $\omega_H(u,v) \ge d_G(u,v)$, where $d_G(u,v)$ stands for the {\em distance} between $u$ and $v$ in $G$. For a positive integer parameter $\beta$, and a positive parameter $\eps > 0$, the graph $G_H$ is called a $(\beta,\eps)$-{\em hopset} of $G$, if for every pair of vertices $u,v \in V$, the {\em $\beta$-bounded distance} $d^{(\beta)}_{G \cup G_H}(u,v)$ between $u$ and $v$ in $G \cup G_H$ is within a factor $1 + \eps$ from the distance $d_G(u,v)$ between these vertices in $G$, i.e.,
$d_G(u,v)  \le  d_{G \cup G_H}^{(\beta)}(u,w) \le (1+ \eps) d_G(u,v)$.
The {\em $\beta$-bounded distance} between $u$ and $v$ in $G$ is the length of the shortest {\em $\beta$-bounded $u$-$v$-path}  $\Pi_{u,v}$ in $G$, i.e., of a shortest path $\Pi_{u,v}$ with at most $\beta$ hops/edges.

Hopsets constitute an important algorithmic and combinatorial object, and they are extremely useful for approximate distance computations in a large variety of computational settings, including the distributed model \cite{N14,HKN16,EN16,EN16a,E17}, parallel (PRAM) model \cite{KS97,SS99,C00,MPVX15,EN16}, streaming model \cite{N14,HKN16,EN16,E17},  dynamic setting \cite{B09,HKN14}, and for routing \cite{LP15,EN16}.
The notion of hopsets was coined in a seminal paper of Cohen \cite{C00}. (Though some first implicit constructions appeared a little bit earlier \cite{UY91,KS97,SS99,C97}.

 In \cite{C00}, Cohen also devised landmark constructions of hopsets. Specifically, she showed that for any parameters $\kappa = 1,2,\ldots$, and $\eps > 0$, and any $n$-vertex graph $G$, there exists a $(\beta,\eps)$-hopset with $O(n^{1+1/\kappa} \cdot \log n)$ edges, with $\beta = O\left({{\log n} \over \eps}\right)^{O(\log \kappa)}$. Moreover, she showed that hopsets with comparable attributes can be efficiently constructed in the centralized and parallel settings.
Specifically, in the centralized setting her algorithm takes an additional parameter $\rho > 0$, and constructs a hopset of size $O(n^{1+1/\kappa} \cdot \log n)$ edges in $O(|E| n^\rho)$ time, with hopbound $\beta = O\left({{\log n} \over \eps}\right)^{{{O(\log \kappa)} \over \rho}}$. In the PRAM model, her hopset has size $O(n^{1+1/\kappa}) \cdot O(\log n)^{O({{\log \kappa} \over \rho})}$, and $\beta$ is as above, and it is constructed in roughly $O(\beta)$ (i.e., polylogarithmic time), and with $O(|E| \cdot n^\rho)$ work.

Cohen  \cite{C00} also raised the open question of existence and efficient constructability of hopsets with better attributes; she called it an ``intriguing research problem". In the two decades that passed since Cohen's work \cite{C00}, numerous algorithms for constructing hopsets
in various settings were devised \cite{B09,HKN14,N14,MPVX15,EN16}.
The hopsets of \cite{B09,HKN14,N14,HKN16} are no better than those of \cite{C00} in terms of their attributes. (But they are constructed in settings to which Cohen's algorithm is not known to apply.) The algorithm of \cite{MPVX15} builds hopsets of size $O(n)$, but with a large hopbound $\beta = n^{\Omega(1)}$.
In \cite{EN16}, the current authors showed that there always exist $(\beta,\eps)$-hopsets of size $O(n^{1+1/\kappa} \cdot \log n)$, with {\em constant} (i.e., independent of $n$) hopbound $\beta = O({{\log \kappa} \over \eps})^{\log \kappa +O(1)}$.
Abboud et al. \cite{ABP17} showed a lower bound $\beta = \Omega({1 \over {\eps \log \kappa}})^{\log \kappa}$ on the hopbound of hopsets with size $O(n^{1+1/\kappa})$.

In the PRAM model, \cite{EN16} showed two results. First, that for parameters $\eps,\rho,\zeta > 0$ and $\kappa = 1,2,\ldots$, hopsets with constant hopbound
$\beta = O\left({{\log \kappa + 1/\rho} \over {\eps \zeta}}\right)^{\log \kappa + 2/\rho + O(1)}$ can be constructed in time $O(n^\zeta \cdot \beta)$, using $O(|E| \cdot n^{\rho + \zeta})$ work. Second, \cite{EN16} devised a PRAM algorithm with polylogarithmic time
$O\left({{\log n} \over \eps}\right)^{\log \kappa + 1/\rho + O(1)}$, albeit with a polylogarithmic $\beta$ roughly equal to the running time.
The hopset's size is $O(n^{1+1/\kappa} \cdot \log n)$ in the second result as well.

These results of \cite{EN16} strictly outperformed the longstanding tradeoff of \cite{C00} in all regimes, and proved existence of hopsets with constant hopbound. However, they left a significant room for improvement.
 First, the hopsets of \cite{EN16}  have $\Omega(n \log n)$ edges in all regimes. As a result, the only currently existing sparser hopsets \cite{UY91,KS97,SS99,MPVX15} have hopbound of $n^{\Omega(1)}$. Hence it was left open in \cite{EN16} if hopsets with size $o(n \log n)$ and hopbound $n^{o(1)}$ exist.
Second, the question  whether hopsets with constant hopbound can be constructed in polylogarithmic PRAM time was left open in \cite{EN16}. Indeed, the hopsets' algorithm of \cite{EN16} for constructing such hopsets requires $n^{\Omega(1)}$ PRAM time.

In this paper we answer both these questions in the affirmative. Specifically, we show that for any $\kappa = 1,2,\ldots$,   there exists a $(\beta(\kappa,\eps),\eps)$-hopset, for all $\eps > 0$ {\em simultaneously}, with size $O(n^{1+1/\kappa})$ and $\beta = \beta(\kappa,\eps) = O\left({{\log \kappa } \over \eps}\right)^{\log \kappa + O(1)}$.
In particular, by setting $\kappa = \log n$, we obtain a {\em linear-size} hopset, and its hopbound is $\beta = O\left({{\log \log n} \over\eps}\right)^{\log\log n + O(1)}$. This is an almost exponential improvement of the previously best known upper bound (due to \cite{MPVX15}) on the hopbound of linear-size hopsets.

Second, in the PRAM setting, for any $\kappa = 1,2,\ldots$, $\eps > 0$, and $\rho > 0$, our algorithm constructs hopsets of size
$O(n^{1+1/\kappa} \cdot \log^* n)$, {\em constant hopbound} $\beta = O\left({{(\log \kappa + 1/\rho})^2 \over \eps}\right)^{\log \kappa + 1/\rho + O(1)}$, in {\em polylogarithmic} time $O((\log n)/\eps)^{\log\kappa + 1/\rho + O(1)}$, using work $O(|E| \cdot n^\rho)$.
This is an exponential improvement of the parallel running time of the previously best-known algorithm for constructing hopsets with constant hop-bound due to \cite{EN16}.
We can also shave the $\log^* n$ factor in the size, that allows for a linear-size hopset, but then $\beta$ grows to be roughly the running time.
See Table \ref{table:pram} for a concise comparison between existing and new results concerning hopsets in the PRAM model.

Our algorithm also provides improved results for constructing hopsets in distributed CONGEST and Congested Clique models (see \sectionref{sec:dist} for definition of these models). In all these models, our algorithm constructs linear-size hopsets.
Also, the running time of our algorithms in all these models is purely combinatorial, i.e., it does not depend on the aspect ratio $\Lambda$ of the graph. \footnote{The aspect ratio $\Lambda$ of a graph $G$ is given by $\Lambda = {{\max_{u,v \in V} d_G(u,v)} \over {\min_{u,v\in V, u\neq v} d_G(u,v)}}$.} In contrast, previous algorithms \cite{N14,HKN16,EN16} for constructing hopsets in the CONGEST model all have running time proportional to $\log \Lambda$.

\begin{table*}[h]
\centering
\small
\begin{tabular}{ |l|l|l|l|l| }
\hline
 Reference & Size & $\beta$ = Hopbound & Time & Work \\
\hline
\cite{KS97,SS99} & $O(n)$ & $O(\sqrt{n})$ & $O(\sqrt{n}\log n)$ & $O(|E|\cdot\sqrt{n})$\\
\hlineB{2}
\multirow{2}{*}{\cite{MPVX15}} & $O(n)$ & $O(n^{\frac{4+\alpha}{4+2\alpha}})$ & $O(n^{\frac{4+\alpha}{4+2\alpha}})$ & $O(|E|\cdot\log^{3+\alpha}n)$\\
\cline{2-5}
& $O(n)$ & $O(n^\alpha)$ ($\alpha\ge\Omega(1)$)& $O(n^{\alpha})$ & $O(|E|\cdot\log^{O(1/\alpha)}n)$\\
\hlineB{2}
\cite{C00} & $n^{1+1/\kappa}\cdot (\log n)^{O(\frac{\log\kappa}{\rho})}$ & $(\log n)^{O(\frac{\log\kappa}{\rho})}$ & $(\log n)^{O(\frac{\log\kappa}{\rho})}$  &  $O(|E|\cdot n^\rho)$ \\
\hlineB{2}
\multirow{2}{*}{\cite{EN16}}& $O(n^{1+\frac{1}{\kappa}}\cdot \log n)$ & $\left(\log n\right)^{\log\kappa+\frac{1}{\rho} + O(1)}$ & $O(\beta)$  & $O(|E|\cdot n^\rho)$ \\
\cline{2-5}
& $O(n^{1+\frac{1}{\kappa}}\cdot\log n)$ & $\left(\frac{\log\kappa+\frac{1}{\rho}}{\zeta}\right)^{\log\kappa+O(\frac{1}{\rho}) }$ & $O(n^\zeta\cdot\beta)$  & $O(|E|\cdot n^{\rho+\zeta})$ \\
\hlineB{2}
\multirow{2}{*}{\bf This paper}&$O(n^{1+\frac{1}{\kappa}})$ & $\left(\log n\right)^{\log\kappa+\frac{1}{\rho} + O(1)}$ & $O(\beta)$  & $O(|E|\cdot n^\rho)$ \\
\cline{2-5}
& $O(n^{1+\frac{1}{\kappa}}\cdot \log^*n)$ & $\left(k+\frac{1}{\rho}\right)^{O(\log\kappa+\frac{1}{\rho})}$ & $(\log n)^{\log\kappa+\frac{1}{\rho} + O(1)}$  & $O(|E|\cdot n^\rho)$ \\
\hline

\end{tabular}
\caption{Comparison between $(\beta,\epsilon)$-hopsets in the PRAM model (neglecting the dependency on $\epsilon$). The hopsets of \cite{KS97,SS99} provide exact distances. 
}\label{table:pram}
\end{table*}

\subsection{Distributed Routing with Small Memory}

The main application of our novel hopsets' construction is to the problem of distributed construction of compact routing schemes. A {routing} scheme has two main phases: the preprocessing phase, and the routing phase. In the preprocessing phase, each vertex is assigned a routing table and a routing label.\footnote{In this paper we only consider {\em labeled} or {\em name-dependent} routing, in which vertices are assigned labels by the scheme. There is also a large body of literature on name-independent routing schemes; cf. \cite{AGMNT08} and the references therein. However, a lower bound \cite{LP13} shows that constructing a name-independent routing scheme with stretch $\rho$ requires $\tilde{\Omega}(n/\rho^2)$ time in the CONGEST model. See also \cite{GGHI13} for lower bounds on the {\em communication} complexity of the preprocessing phase of distributed routing.}
In the routing phase, a vertex $u$ gets a message $M$ with a short header $\mathit{Header}(M)$ and with a destination label $\mathit{Label}(v)$ of a vertex $v$, and based on its routing table $\mathit{Table}(u)$, on $\mathit{Label}(v)$, and on $\mathit{Header}(M)$,  the vertex $u$ decides to which neighbor $x \in \Gamma(u)$ to forward the message $M$, and which header to attach to the message. The {\em stretch} of a routing scheme is the worst-case ratio between the length of a path on which a message $M$ travels, and the graph distance between the message's origin and destination.

Due to its both theoretical and practical appeal, routing is a central problem in distributed graph algorithms \cite{PU89,ABLP90,TZ01-spaa,C01,EGP03,GP03,AGM04,C13}. A landmark routing scheme was devised in \cite{TZ01-spaa}. For an integer $k \ge 1$, the stretch of their scheme is $4k-5$, the tables are of size $O(n^{1/k} \log^{1-1/k} n)$,  the labels are of size $O(k  \log n)$, and the headers are of size $O(\log n)$. Chechik \cite{C13} improved this result, and devised a scheme with stretch $3.68 k$, and other parameters like in \cite{TZ01-spaa}.

An active thread of research \cite{ABLP90,AP92,LP13,GGHI13,LP15,EN16a} focuses on efficient implementation of the {\em preprocessing} phase  of routing in the distributed CONGEST model, i.e., computing compact tables and short labels that enable for future low-stretch routing. This problem was raised in a seminal paper by Awerbuch, Bar-Noy, Linial and Peleg \cite{ABLP90}, who devised a routing scheme with stretch $2^{O(k)}$, overall memory requirement $\tO(n^{1+1/k})$,\footnote{$\tilde{O(f(n)}$-notation hides polylogarithmic in $f(n)$ factors.}
individual memory requirement for a vertex $v$ of $\tO(\deg(v) + n^{1/k})$, and construction time $\tO(n^{1+1/k})$ (in the CONGEST model).  The ``individual memory requirement" parameter encapsulates the routing tables and labels, and the memory used while computing the tables and labels.

Lenzen and Patt-Shamir \cite{LP15} devised a distributed routing scheme (based on \cite{TZ01-spaa}) with stretch $4k-3 +o(1)$, tables of size $O(n^{1/k} \log n)$, labels of size $O(k \log n)$, individual memory requirement of $\tO(n^{1/k})$, and construction time $\tO(S+n^{1/k})$, where $S$ is the {\em shortest-path diameter} of the input graph $G$, i.e., the maximum number of hops in a shortest path between a pair of vertices in $G$. Though $S$ is often much smaller than $n$,
it is desirable to evaluate complexity measures of distributed algorithms in terms of $n$ and $D$, where $D$ is the {\em hop-diameter} of $G$, defined as the maximum distance between a pair of vertices $u,v$ in the underlying unweighted graph of $G$.
 Typically, we have $D \ll S \ll n$, and it is always the case that $D \le S \le n$.
(See Peleg's book \cite{P00} for a comprehensive discussion.)

Lenzen and Patt-Shamir \cite{LP13} also devised a routing scheme with tables of size $\tilde{O}(n^{1/2+1/k})$, labels of size $O(\log n\cdot\log k)$, stretch at most $O(k\log k)$, and has running time of $\tilde{O}(n^{1/2+1/k}+D)$ rounds.
They (based on \cite{DHK12}) also showed a lower bound of $\tilde{\Omega}(D + \sqrt{n})$ on the time required to construct  a routing scheme.
In a follow-up paper, \cite{LP15} showed how to improve the stretch of the above scheme to $O(k)$. The main drawback of this result is the prohibitively large size of the routing tables. (The individual memory requirement is consequently prohibitively large as well.)
They also exhibited a different tradeoff, that overcame the issue of large routing tables. They devised an algorithm that produced routing tables of size $O(n^{1/k}\cdot\log^2n)$, labels of size $O(k\log^2n)$ and stretch $4k-3+o(1)$,
 albeit with sub-optimal running time $\tO(\min\{ (nD)^{1/2} n^{1/k},n^{2/3} +D\}) \cdot \log \Lambda$, and no guarantee on the individual memory requirement during the preprocessing phase. In \cite{EN16a}, the current authors improved the bounds of \cite{LP13,LP15}. In the current state-of-the-art scheme \cite{EN16a}, the stretch is $4k-5 + o(1)$, the tables and labels are of the same size as in \cite{LP13,LP15}
(i.e., $O(n^{1/k} \log^2 n)$ and $O(k \log^2 n)$, respectively), the construction time is 
 $O((n^{1/2 + 1/k}  +D) \cdot \min \{(\log n)^{O(k)},2^{\tO(\sqrt{\log n})}\} \cdot \log \Lambda$. (A similar, though slightly weaker, result was achieved by \cite{LPP16}.)  Still there is no meaningful guarantee on the individual memory requirement in the preprocessing phase.
See Table \ref{fig:table} for a concise summary of existing bounds, and a comparison with our new results.


To summarize, all currently existing distributed routing algorithms with nearly-optimal running time  $\approx D + n^{1/2+1/k}$ suffer from three issues.
First, they provide no meaningful guarantee on individual memory requirement on vertices in the preprocessing phase; second, their preprocessing time is not purely combinatorial, but rather depends linearly on $\log \Lambda$; and third, their tables and labels sizes are roughly  $O(\log n)$ off from the respective tables and labels' sizes of Thorup-Zwick's sequential construction \cite{TZ01-spaa}.

The issue of individual memory requirement was indeed explicitly raised by Lenzen \cite{L16} in a private communication with the authors.
He wrote (the stress on ``during" is in the origin):
\\
\\
``{\em One annoying thing about this is that there is a huge amount of storage required *during* the construction. It seems
odd that the nodes can hold only small tables, but should have large memory during the construction. I think it's an interesting
question whether we can have a good construction using $\tilde{O}(n^{1/k})$ memory only. A hop set may be the wrong option for
this, because reflecting the distance structure of the skeleton accurately cannot be done by a sparse graph; on the other hand,
maybe there's some distributed representation cleverly distributing the information over the graph nodes?}"
\\
\\
Based on our novel hopsets' construction, we devise an algorithm that addresses all these  issues. Specifically, the stretch of our scheme is $4k-5 +o(1)$, the sizes of tables and labels essentially match the respective sizes of Thorup-Zwick's construction, i.e., they are $O(n^{1/k} \log n)$ and $O(k \log n)$, respectively.  Our construction time is $(n^{1/2+1/k} + D)\cdot(\log n)^{O(k)}$, i.e., it is purely combinatorial. Most importantly, the individual memory requirement is at most $\tO(n^{1/k})$.
Moreover, we can reduce the running time to $(n^{1/2+1/k} + D)\cdot \min\{(\log n)^{O(k)}, 2^{\tO(\sqrt{\log n})}\}$, while the individual memory increases slightly to $\max\{\tO(n^{1/k}),2^{\tO(\sqrt{\log n})}\}$. In particular, we can have a polylogarithmic individual memory requirement and construction time $O((n^{1/2} +D) \cdot n^\eps)$, for an arbitrarily small constant $\eps > 0$.

{\bf Distributed Tree Routing:}
An important ingredient in the existing distributed routing schemes \cite{LP15,EN16a} for general graphs, and in our new routing scheme, is a distributed tree routing scheme. Thorup and Zwick \cite{TZ01-spaa} showed that with routing tables of size $O(1)$ and labels of size $O(\log n)$, one can have an {\em exact} (i.e., no stretch) tree routing. \cite{LP15,EN16a} showed that in $\tO(D+ \sqrt{n})$ time, one can construct exact tree routing with tables and labels of size $O(\log n)$ and  $O(\log^2 n)$, respectively, i.e., there is an overhead of $\log n$ in both parameters with respect to Thorup-Zwick's sequential construction. In this paper we improve this result, and devise a $\tO(D+ \sqrt{n})$-time algorithm that constructs tree-routing tables and labels of sizes that match the sequential construction of Thorup and Zwick, i.e., of size $O(1)$ and $O(\log n)$, respectively. Moreover, if one is interested in a scheme that always routes via the root of the tree, as is the case in the application to routing in general networks, then our algorithm for constructing tables and labels that supports this requires only a small ($O(\log n)$) individual memory in each vertex.

\begin{table}
\centering
\footnotesize
\begin{tabular}{|c|c|c|c|c|c|}
	\hline
     Reference  &  Number of Rounds   &   Table size & Label size & Stretch  & Memory per vertex\\
	\hlineB{3}
\cite{ABLP90}   & $O(n^{1+\frac{1}{k}})$ & $O(n^{\frac{1}{k}}\cdot\log n)$ & $O(k\log n)$ & $2\cdot 3^k-1$ &  $\tilde{O}(\deg(v)\!+\!n^{\frac{1}{k}})$\\
\hline
\cite{TZ01-spaa,C13}   & $O(n^{1+\frac{1}{k}})$ & $O(n^{\frac{1}{k}}\cdot\log n)$ & $O(k\log n)$ & $3.68k$ &  $O(n^{\frac{1}{k}}\cdot\log n)$\\
\hline
\cite{LP13,LP15}  & $\tilde{O}(n^{\frac{1}{2}+\frac{1}{4k}}+D)$ & $\tilde{O}(n^{\frac{1}{2}+\frac{1}{4k}})$ & $O(\log n)$ & $6k-1+o(1)$ & $\tilde{O}(n^{\frac{1}{2}+\frac{1}{4k}})$ \\
\hline
\multirow{2}{*}{\cite{LP15}} & $\tilde{O}(S+n^{\frac{1}{k}})$ &  $O(n^{\frac{1}{k}}\cdot\log n)$ & $O(k\log n)$ & $4k-3$ & $O(n^{\frac{1}{k}}\cdot\log n)$\\
\cline{2-6}
 & $\tilde{O}(\min\!\{\!(nD)^{\frac{1}{2}}\!\cdot\! n^{\frac{1}{k}}\!,\!n^{\frac{2}{3}\!+\!\frac{2}{3k}}\!+\! D\!\})$ & $O(n^{\frac{1}{k}}\cdot\log^2n)$ & $O(k\log^2n)$ & $4k-3+o(1)$ & $\tilde{O}(n^{\frac{1}{2}})$\\
 \hline
\cite{LPP16} & $(n^{\frac{1}{2}+\frac{1}{k}}+D)\cdot 2^{\tilde{O}(\sqrt{\log n})}$ & $O(n^{\frac{1}{k}}\cdot\log^2n)$ & $O(k\log^2n)$ & $4k-3+o(1)$ &  $\tilde{O}(n^{\frac{1}{2}})$\\
\hline
\cite{EN16a} & $(n^{\frac{1}{2}+\frac{1}{k}}+D)\cdot \beta$ & $O(n^{\frac{1}{k}}\cdot\log^2n)$ & $O(k\log^2n)$ & $4k-5+o(1)$ &  $\tilde{O}(n^{\frac{1}{2}})$\\
\hlineB{2}
\multirow{2}{*}{\bf This paper} & $(n^{\frac{1}{2}+\frac{1}{k}}+D)\cdot(\log n)^{O(k)}$ & $O(n^{\frac{1}{k}}\cdot\log n)$ & $O(k\log n)$ & $4k-5+o(1)$ &  $\tilde{O}(n^{\frac{1}{k}})$\\
\cline{2-6}
& $(n^{\frac{1}{2}+\frac{1}{k}}+D)\cdot 2^{\tilde{O}(\sqrt{\log n})}$ & $O(n^{\frac{1}{k}}\cdot\log n)$ & $O(k\log n)$ & $4k-5+o(1)$ &  $2^{\tilde{O}(\sqrt{\log n})}$\\
	\hline
\end{tabular}
\caption{Comparison of compact routing schemes for graphs with $n$ vertices, $m$ edges, hop-diameter $D$, and shortest path diameter $S$. Denote $\beta = \min\{(\log n)^{O(k)},2^{\tilde{O}(\sqrt{\log n})}\}$.}\label{fig:table}
\end{table}


\subsection{Technical Overview}
\label{sec:tech}

Cohen's algorithm \cite{C00} is based on  a subroutine for constructing pairwise covers \cite{C93,ABCP93}, i.e., collections of small-radii clusters with small maximum overlap (no vertex belongs to too many clusters). The algorithm is a top-down recursive procedure: it interconnects large clusters of the cover via hopset edges, and recurses in small clusters. To keep the overall overlap of all recursion levels in check, Cohen used the radius parameter $O(\log n)$ for the covers. This resulted in a hop-bound, which is at least polylogarithmic in $n$.  Cohen's hopset is also built separately for each distance scale $[2^i,2^{i+1})$, $i = 0,1,\ldots,\log \Lambda$, and the ultimate hopset is the union of all these single-scale hopsets.

The hopset's construction of \cite{EN16}, (due to the current authors), also builds a single-scale hopset for each distance scale, and then takes their union as an ultimate hopset. The construction of single-scale hopsets in \cite{EN16} is based upon ideas from the construction of $(1+\eps,\beta)$-spanners of \cite{EP04} for unweighted graphs. It starts with a partition of the vertex set $V$ into singleton clusters $\cP_0 = \{v\} \mid v \in V\}$, and alternates superclustering and interconnection steps. In a superclustering step some of the clusters of the current partition are merged into larger clusters (of $\cP_{i+1}$), while the other clusters are interconnected with one another via hopset edges.

This approach (of \cite{EN16}) enabled us to prove existence of hopsets with constant hop-bound, but it appears to be uncapable of producing hopsets of size $o(n \log n)$.  Indeed, even if each single-scale hopset is of linear size (which is indeed the case in \cite{EN16}), their union is doomed to be of size $\Omega(n \log n)$. Moreover, in parallel and distributed settings, one produces a hopset $i+1$ based upon a hopset of scale $i$. This results in {\em accumulation of stretch} from $1+ \eps$ to $(1+ \eps)^{\log n}$. To alleviate this issue, one needs to rescale $\eps$.
However,  then the hopbound grows from constant to polylogarithmic. To get around this, \cite{EN16} used a smaller number of scales, and this indeed enables \cite{EN16} to construct hopsets with constant hopbound in these settings, albeit the running time becomes proportional to the ratio between consequent scales, i.e., it becomes $n^{\Omega(1)}$.

The closest to our current construction of hopsets is the line of research of \cite{B09,HKN14,HKN16,N14}, which is based on a construction of distance oracles due to Thorup and Zwick \cite{TZ01}.
To construct their oracles, \cite{TZ01} used a hierarchy of sets $V = A_0 \supseteq A_1 \supseteq \ldots A_{\kappa-1} \supseteq A_\kappa = \emptyset$, where each vertex of $A_i$, for all $i = 0,1,\ldots,\kappa-2$, is sampled with probability $n^{-1/\kappa}$ for the inclusion into $A_{i+1}$.
For each vertex $u \in V$,  \cite{TZ01} defined for every $i = 0,1,\ldots,\kappa-1$, the $i$th {\em pivot} $p_i(u)$ to be its closest vertex in $A_i$,   and the  $i$th {\em bunch} $B_i(u) = \{v \mid d_G(u,v) < d_G(u,A_{i+1})\} \cup \{p_{i+1}(u)\}$, (for $i=\kappa$, let $\{p_\kappa(u)\} = \emptyset$),
and the entire bunch $B(u) = \bigcup_{i=0}^{\kappa-1} B_i(u)$. They also defined the dual sets, {\em clusters}, $C(v) = \{u \mid v \in B(u)\}$.
Bernstein and others  \cite{B09,HKN14,N14,HKN16} used this construction with $\kappa = \tilde{\Theta}(\sqrt{\log n})$, and built Thorup-Zwick clusters with respect to $2^{\tO(\sqrt{\log n})}$-bounded distances.
As a result, they obtained a so-called {\em $2^{\tO(\sqrt{\log n})}$-bounded hopset}, i.e., a hopset which takes care only of pairs $u,v \in V$ of vertices that admit a $2^{\tO(\sqrt{\log n})}$-bounded shortest path. They then used this bounded hopset in a certain recursive fashion
(see the so-called {\em hop reduction} of Nanongkai \cite{N14}), to obtain their ultimate hopset.

Thorup-Zwick's construction with $\kappa =\tilde{\Theta}(\sqrt{\log n})$ alone introduces into the hopset $n \cdot 2^{\tilde{\Theta}(\sqrt{\log n})}$ edges, and thus, such a hopset cannot be very sparse. In addition, the recursive application of the hop reduction technique results in a hopbound of $2^{\tilde{\Omega}(\sqrt{\log n})}$.

Our construction of hopsets is based upon a construction of Thorup-Zwick's emulators\footnote{A graph $G' = (V,E',\omega')$ is called a {\em sublinear-error emulator} of an unweighted graph $G = (V,E)$, if for every pair of vertices $u,v \in V$, we have $d_{G}(u,v)  \le d_{G'}(u,v) \le d_G(u,v) + \alpha(d_G(u,v))$ for some sub-linear stretch function $\alpha$.
If $G'$ is a subgraph of $G$, it is called a sublinear-error  {\em spanner} of $G$.}, from a different paper by Thorup and Zwick \cite{TZ06}.
Specifically, to obtain the hierarchy $V = A_0 \supseteq A_1 \supseteq \ldots A_{\log\kappa-1} \supseteq A_{\log\kappa} = \emptyset$, one samples each vertex of $A_i$, for $i=0,1,\ldots,\log\kappa-2$, with probability roughly $n^{-2^i/\kappa}$ for inclusion in $A_{i+1}$.
Then one defines the bunch of a vertex $u \in A_i$ as $B(u) =  \{v \in A_i \mid d_G(v,u) < d_G(v,A_{i+1})\} \cup \{p_{i+1}(u)\}$, and sets
\begin{equation}
\label{eq:emul_hopset}
H ~=~ \bigcup_{u \in V}  \{(u,v) \mid v \in B(u)\}~.
\end{equation}

For {\em unweighted} graphs $G$, \cite{TZ06} showed that $H$ given by (\ref{eq:emul_hopset}) is an additive emulator with stretch $\alpha(d)=O(\log\kappa\cdot d^{1-1/(\log\kappa-1)})$ and
$O(\log\kappa \cdot n^{1+1/\kappa})$ edges
. By a different proof argument, we show that the {\em very same} construction provides also a $(\beta,\eps)$-hopset of the same size and with $\beta=O\left({\log\kappa \over \eps}\right)^{\log\kappa-1}$, for all $\epsilon>0$ {\em simultaneously}.
Moreover, by adjusting the sampling probabilities, we also shave the $\log\kappa$ factor from both the hopset's and the emulator's size, while increasing the exponent of $\beta$ by 1.
(This also gives rise to the first linear-size emulator with sub-linear additive stretch for unweighted graphs.)

As a result, we obtain a construction of hopsets, which is by far simpler than the previous Thorup-Zwick-based constructions of hopsets \cite{B09,HKN14,N14,HKN16}. As was discussed above, it also provides hopsets with much better parameters, and it is more adaptable to efficient implementation in various computational settings. Our construction is also much simpler than the constructions of \cite{C00,EN16}, which are not based on Thorup-Zwick's hierarchy.

Parallel and distributed implementations of our hopset's construction proceed in scales  $\ell = 0,1,\ldots,\log n$, where on scale $\ell$ the algorithm constructs a $2^\ell$-bounded hopset. This is different from the situation in \cite{C00,EN16,N14,HKN16},
 where scale $\ell$ takes care of distances in the range $[2^\ell,2^{\ell+1})$.  An important advantage of this is that we no longer need to take the union of all single-scale hopsets into our  hopset; rather we just take the largest-scale hopset as our ultimate hopset. This saves a factor of $\log n$ in the size, and enables us to construct {\em linear-size} hopsets in parallel and distributed settings. The fact that we do not work with distance scales, but rather with hop-distance-scales, makes it possible to avoid the dependence on $\log \Lambda$ in the distributed construction time, and to achieve a purely combinatorial running time.  All previous distributed algorithms for constructing approximate hopsets \cite{N14,HKN16,EN16} have running time proportional to $\log \Lambda$. (A distributed construction of an {\em exact} hopset \cite{E17} by the first-named author, however, avoids this dependence too. Alas, it has a much higher running time.)

The fact that the construction's scales are with respect to hop-distances, as opposed to actual distances, enables us to essentially avoid accumulation of error. This is done by the following recursive procedure.
First, we build $2^\ell$-bounded hopsets $H^{(\ell)}$, $\ell = 0,1,\ldots, \log n$, and set $H(1) = H^{(\log n)}$ to be the highest-scale hopset. This process involves accumulation of stretch, and after rescaling the stretch parameter $\eps$, the hopset $H(1)$ ends up having polylogarithmic hopbound $\beta_1$. Its construction time is roughly $\beta_1$, i.e., polylogarithmic as well. Now we add the hopset into the original graph, and recurse on $G \cup H(1)$. Note that now we only need to process $\log \beta_1 \approx \log\log n$ scales, rather than $\log n$ ones. Hence the accumulation of stretch in the resulting hopset $H(2)$ is much more mild than in $H(1)$. As a result, after rescaling the stretch parameter $\eps$, the hopbound $\beta_2$ of $H(2)$ is roughly $\mathit{poly}(\log\log n)$. By repeating this recursive process for $\log^* n$ iterations, we eventually achieve a hopset with constant hop-bound in parallel polylogarithmic time. As was mentioned above, this dramatically improves the $n^{\Omega(1)}$ parallel time required in \cite{EN16} to construct a hopset with constant hopbound.

In the context of distributed routing, like in \cite{LP15,EN16a}, our hopset $H$ is constructed  on top of a {\em virtual} (aka skeleton) graph $G' = (V',E')$, where $V'$ is a collection of $\approx \sqrt{n}$ vertices, sampled from the original vertex set $V$ independently with probability $\approx n^{-1/2}$. There is an edge $(u',v') \in E'$ iff there is a $\tilde{O}(\sqrt{n})$-bounded $u'-v'$ path in $G$. However, since we aim to design an algorithm in which vertices employ only a small memory during the preprocessing phase, we cannot afford computing the virtual graph $G'$. Rather, somewhat surprisingly, we show that the hopset $H$ for $G'$ can be constructed without ever constructing $G'$ itself! We only compute those edges of $G'$, which are required for constructing the hopset $H$.

Note, however, that unlike a spanner or a low-stretch spanning tree, hopset is always used {\em in conjunction} with the graph for which it was constructed. In other words, to compute the Thorup-Zwick routing scheme for the virtual graph $G'$, we conduct Bellman-Ford explorations in $G' \cup H$. So, at the first glance, it seems necessary to eventually compute $G'$, for being able to conduct these Bellman-Ford explorations.

We cut this gordian knot by computing only those edges of $G'$ that are really needed for computing either the hopset $H$ or the TZ routing scheme for $G' \cup H$. This turns out to be (typically) a small fraction of edges of $G'$, and those edges can be computed much more efficiently than the entire $G'$, and using much smaller memory. This idea also enables us to compute lengths of these edges of $G'$ {\em precisely}, as opposed to approximately, as it was done in \cite{LP15,HKN16,EN16a}. This simplifies the analysis of the resulting scheme.
We note that the idea of computing a hopset $H$ without first computing the underlying virtual graph $G'$, and conducting Bellman-Ford explorations in $G' \cup H$ without ever computing $G'$ in its entirety appeared in a recent work \cite{E17}, by the first-named author.  \cite{E17} constructed an exact hopset of \cite{SS99,N14} with a polynomial hopbound. However, the exact hopset is a much simpler structure than the small-hopbound approximate hopset that we construct here. Showing that this idea is applicable for our new approximate small-hopbound hopset is technically  substantially more challenging.

Another crucial idea that we employ to guarantee a small individual memory requirement is ensuring that our hopset $H$ has small {\em arboricity}, i.e., that its edges can be oriented in such a way that every vertex has only a small out-degree. This out-degree is proportional to the ultimate individual memory requirement of our algorithm.

\subsection{Organization}
\label{sec:org}
Our linear-size hopsets appear in \sectionref{sec:hop}, and the distributed construction in \sectionref{sec:CC} for Congested Clique and \sectionref{sec:CONGEST} for the CONGEST model. The hopsets in the PRAM model are presented in \sectionref{sec:pram}. Finally, in \sectionref{sec:route} we describe our distributed routing scheme with small memory, and the distributed tree routing in \sectionref{sec:tree-route}.

\section{Linear Size Hopsets}
\label{sec:hop}

Let $G=(V,E)$ be a weighted graph, and fix a parameter $k\ge 1$. Let $\nu=1/(2^k-1)$ (one should think of $\kappa=1/\nu$). Let $V=A_0\supseteq A_1\supseteq\dots\supseteq A_k=\emptyset$ be a sequence of sets, such that for all $0\le i<k-1$, $A_{i+1}$ is created by sampling every element from $A_i$ independently with probability $p_i=n^{-2^i\cdot\nu}$.\footnote{Our definition is slightly different than that of \cite{TZ06}, which used $p_i=|A_i|/n^{1+\nu}$, but it gives rise to the same expected size of $A_i$. We use our version since it allows efficient implementation in various models of computation.} It follows that for $0\le i\le k-1$ we have
\[
N_i:=\E[|A_i|]=n\cdot\prod_{j=0}^{i-1}p_j=n^{1-(2^i-1)\nu}~,
\]
and in particular $N_{k-1}=n^{(1+\nu)/2}$.

For every $0\le i\le k-1$ and every vertex $u\in A_i\setminus A_{i+1}$, define the pivot $p(u)\in A_{i+1}$ as a vertex satisfying $d_G(u,A_{i+1})=d_G(u,p(u))$ (note $p(u)$ does not exist for $u\in A_{k-1}$), and define the bunch
$$B(u)=\{v\in A_i~:~ d_G(u,v)<d_G(u,A_{i+1})\}\cup\{p(u)\}~.$$ 
The hopset is created by taking $H=\{(u,v)~:~u\in V, v\in B(u)\}$, where the length of the edge $(u,v)$ is set as $d_G(u,v)$. As argued in \cite{TZ01}, for any $0\le i\le k-2$ and $u\in A_i\setminus A_{i+1}$, the size of $B(u)$ is bounded by a random variable sampled from a geometric distribution with parameter $p_i$ (this corresponds to the first vertex of $A_i$, when ordered by distance to $u$, that is included in $A_{i+1}$). Hence $\E[|B(u)|]\le 1/p_i=n^{2^i\cdot\nu}$. For $u\in A_{k-1}$ we have $\E[|B(u)|]=N_{k-1}=n^{(1+\nu)/2}$. The expected size of the hopset $H$ is at most
\[
\sum_{i=0}^{k-2}(N_i\cdot n^{2^i\cdot\nu}) +N_{k-1}\cdot N_{k-1}= k\cdot n^{1+\nu}~.
\]

The following lemma bounds the number of hops and stretch of $H$. Recall that $d^{(t)}_G(u,v)$ is the length of the shortest-path between $u,v$ in $G$ that consists of at most $t$ edges.
\begin{lemma}\label{lem:main}
Fix any $0<\delta<1/(8k)$ and any $x,y\in V$. Then for every $0\le i\le k-1$, at least one of the following holds:
\begin{enumerate}
\item $d_{G\cup H}^{((3/\delta)^i)}(x,y)\le (1+8\delta i)\cdot d_G(x,y)$.
\item There exists $z\in A_{i+1}$ such that $d_{G\cup H}^{((3/\delta)^i)}(x,z)\le 2d_G(x,y)$.
\end{enumerate}
\end{lemma}
\begin{proof}
The proof is by induction on $i$. We start with the basis $i=0$. If it is the case that $y\in B(x)$, then we added the edge $(x,y)$ to the hopset, i.e. $d_H^{(1)}(x,y)=d_G(x,y)$, and so the first item holds. Otherwise, consider the case that $x\in A_1$: then we can take $z=x$, so the second item holds trivially. The remaining case is that $x\in A_0\setminus A_1$ and $y\notin B(x)$, so by definition of $B(x)$ we get that $d_G(x,y)\ge d_G(x,A_1)$. By taking $z=p(x)$, there is a single edge between $x,z$ in $H$ of length $d_G(x,z)=d_G(x,A_1)\le d_G(x,y)$, which satisfies the second item.

Assume the claim holds for $i$, and we prove for $i+1$. Consider the path $\pi(x,y)$ between $x,y$ in $G$, and partition it into $J\le 1/\delta$ segments $\{L_j=[u_j,v_j]\}_{j\in[J]}$, each of length at most $\delta\cdot d_G(x,y)$, and at most $1/\delta$ edges $\{(v_j,u_{j+1})\}_{j\in[J]}$ of $G$ between these segments. This can be done as follows: define $u_1=x$, and for $j\ge 1$, walk from $u_j$ on $\pi(x,y)$ (towards $y$) until the point $v_j$, which is the vertex so that the next edge will take us to distance greater than $\delta\cdot d_G(x,y)$ from $u_j$ (or until we reached $y$). By definition, $d_G(u_j,v_j)\le\delta\cdot d_G(x,y)$. Define $u_{j+1}$ to be the neighbor of $v_j$ on $\pi(x,y)$ that is closer to $y$ (if exists). If $u_{j+1}$ does not exist (which can happen only when $v_j=y$) then define $u_{j+1}=y$ and $J=j$. Observe that for all $1\le j\le J-1$, $d_G(u_j,u_{j+1})>\delta\cdot d_G(x,y)$, so indeed $J\le 1/\delta$.

We use the induction hypothesis for all pairs $(u_j,v_j)$ with parameter $i$. Consider first the case that the first item holds for all of them, that is, $d_{G\cup H}^{((3/\delta)^i)}(u_j,v_j)\le (1+8\delta i)\cdot d_G(u_j,v_j)$. Then we take the path in $G\cup H$ that consists of the $(3/\delta)^i$-hops between each pair $u_j,v_j$, and the edges $(v_j,u_{j+1})$ of $G$. Since $(3/\delta)^{i+1}\ge (1/\delta)\cdot(3/\delta)^i+1/\delta$, we have
\[
d_{G\cup H}^{((3/\delta)^{i+1})}(x,y)\le \sum_{j\in[J]}(d_{G\cup H}^{((3/\delta)^i)}(u_j,v_j)+d_G^{(1)}(v_j,u_{j+1}))\le(1+8\delta i)\cdot d_G(x,y)~.
\]

The second case is that there are pairs $(u_j,v_j)$ for which only the second item holds. Let $l\in [J]$ (resp., $r\in[J]$) be the first (resp., last) index for which the first item does not hold for the pair $(u_l,v_l)$ (resp., $(u_r,v_r)$). Then there are $z_l,z_r\in A_{i+1}$ such that
\begin{equation}\label{eq:zlzr}
d_{G\cup H}^{((3/\delta)^i)}(u_l,z_l)\le 2d_G(u_l,v_l)~\textrm{ and }~d_{G\cup H}^{((3/\delta)^i)}(v_r,z_r)\le 2d_G(u_r,v_r).
\end{equation}
(Note that we used $v_r$ and not $u_r$ in the second inequality.  This can be done since the lemma's assertion holds for the pair $(v_r,u_r)$ as well, and as the first item is symmetric with respect to $u_r,v_r$, it does not hold for the pair $(v_r,u_r)$ as well.)
Consider now the case that $z_r\in B(z_l)$. In this case we added the edge $(z_l,z_r)$ to the hopset, and by the triangle inequality,
\begin{equation}\label{eq:eq1}
d_H^{(1)}(z_l,z_r)=d_G(z_l,z_r)\le d^{((3/\delta)^i)}_{G\cup H}(u_l,z_l)+d_G(u_l,v_r)+d^{((3/\delta)^i)}_{G\cup H}(z_r,v_r)~.
\end{equation}
Next, apply the inductive hypothesis on segments $\{L_j\}$ for $j<l$ and $j>r$, and in between use the detour via $u_l,z_l,z_r,v_r$. Since $l\le r$, there is at least one segment we skipped, so the total number of hops is bounded by $(1/\delta-1)\cdot(3/\delta)^i+1/\delta+2(3/\delta)^i+1$.
(The additive term of $1/\delta$ accounts for the edges $(v_j,u_{j+1})$, $0 \le j \le J-1$.)
This expression is at most $(3/\delta)^{i+1}$ whenever $\delta<1/2$. It follows that
\begin{eqnarray}\label{eq:case2}
\lefteqn{d_{G\cup H}^{((3/\delta)^{i+1})}(x,y)}\\
&\le& \sum_{j=1}^{l-1}\left[d_{G\cup H}^{((3/\delta)^i)}(u_j,v_j)+d_G^{(1)}(v_j,u_{j+1})\right]+d^{((3/\delta)^i)}_{G\cup H}(u_l,z_l)+ ~d_H^{(1)}(z_l,z_r)\nonumber\\\nonumber
&&+ d^{((3/\delta)^i)}_{G\cup H}(z_r,v_r)+d_G^{(1)}(v_r,u_{r+1})+ \sum_{j=r+1}^{J}\left[d_{G\cup H}^{((3/\delta)^i)}(u_j,v_j)+d_G^{(1)}(v_j,u_{j+1})\right]\\\nonumber
&\stackrel{\eqref{eq:eq1}}{\le}& (1+8\delta i)d_G(x,u_l)+d_G(u_l,v_r)+(1+8\delta i)d_G(v_r,y)\\\nonumber
&&+2d^{((3/\delta)^i)}_{G\cup H}(u_l,z_l)+2d^{((3/\delta)^i)}_{G\cup H}(z_r,v_r)\\\nonumber
&\stackrel{\eqref{eq:zlzr}}{\le}&8\delta\cdot d_G(x,y)+(1+8\delta i)d_G(x,y)\\\nonumber
&=&(1+8\delta(i+1))\cdot d_G(x,y)~.
\end{eqnarray}
This demonstrates item 1 holds in this case. The final case to consider is that $z_r\notin B(z_l)$. Assume first that $z_l\notin A_{i+2}$. Then taking $z=p(z_l)\in A_{i+2}$, the definition of $B(z_l)$ implies that $d_G(z_l,z)\le d_G(z_l,z_r)$. We now claim that item 2 holds for such a choice of $z$. Indeed, since $(z_l,z)\in H$, we have
\begin{eqnarray*}
d_{G\cup H}^{((3/\delta)^{i+1})}(x,z)&\le&\sum_{j=1}^{l-1}\left[d_{G\cup H}^{((3/\delta)^i)}(u_j,v_j)+d_G^{(1)}(v_j,u_{j+1})\right]+d^{((3/\delta)^i)}_{G\cup H}(u_l,z_l)+d_H^{(1)}(z_l,z)\\
&\le&(1+8\delta i)d_G(x,u_l)+d^{((3/\delta)^i)}_{G\cup H}(u_l,z_l) + d_G(z_l,z_r)\\
&\stackrel{\eqref{eq:eq1}}{\le}&(1+8\delta i)d_G(x,u_l)+2d^{((3/\delta)^i)}_{G\cup H}(u_l,z_l)+d_G(u_l,v_r)+ d^{((3/\delta)^i)}_{G\cup H}(z_r,v_r)\\
&\stackrel{\eqref{eq:zlzr}}{\le}&(1+8\delta i)d_G(x,v_r)+6\delta\cdot d_G(x,y)\\
&\le&2d_G(x,y),
\end{eqnarray*}
where the last inequality used that $\delta<1/(8k)$. The case that $z_l\in A_{i+2}$ is simpler, since we may take $z=z_l$.
\end{proof}

Fix any $0<\epsilon<1$, and apply the lemma on any pair $x,y$ with $\delta=\epsilon/(8k)$ and $i=k-1$. It must be that the first item holds (since $A_k=\emptyset$). Hence we have that
\[
d_{G\cup H}^{(\beta)}(x,y)\le (1+\epsilon)d_G(x,y)~,
\]
where the number of hops is given by $\beta=(24k/\epsilon)^{k-1}$. We derive the following theorem.

\begin{theorem}
For any weighted graph $G=(V,E)$ on $n$ vertices, and any $k\ge 1$, there exists $H$ of size at most $O(k\cdot n^{1+1/(2^k-1)})$, which is a $(\beta,\epsilon)$-hopset for any $0<\epsilon<1$ with $\beta=O(k/\epsilon)^{k-1}$.
\end{theorem}

\subsection{Improved Hopset Size}\label{sec:smaller-size}
Here we show how to remove the $k$ factor from the hopset size, at the cost of increasing the exponent of $\beta $ by an additive 1. Note that we may assume w.l.o.g that $k\le\log\log n-1$, as for larger values of $k$, both $\beta$ and the size of the hopset (which becomes $O(kn)$), grow with $k$.
We will increase the number of sets by 1, and sample $V=A'_0\supseteq A'_1\supseteq\dots\supseteq A'_{k+1}=\emptyset$ using the following probabilities: $p'_i=n^{-2^i\cdot\nu}\cdot 2^{2^i-1}$ (the restriction on $k$ ensures $p'_i<1$). Now for $0\le i\le k$,
\[
N'_i:=\E[|A'_i|]=n\cdot\prod_{j=0}^{i-1}p'_j=n^{1-(2^i-1)\nu}\cdot 2^{2^i-i-1}~,
\]
and in particular $N'_k\le 2^{2^k-k}\le n^{1/2}$. The expected size of $H$ becomes at most
\[
\sum_{i=0}^{k-1}(N'_i/p'_i) +N_k'\cdot N_k'\le \sum_{i=0}^{k-1} (n^{1+\nu}/2^i)+n\le 3n^{1+\nu}~.
\]
The hopset construction and the stretch analysis in \lemmaref{lem:main} remains essentially the same. There is an additional sampled set now, and thus  the exponent of $\beta$ grows by an additive 1.
\begin{theorem}
For any weighted graph $G=(V,E)$ on $n$ vertices and any $k\ge 1$, there exists $H$ of size at most $O(n^{1+1/(2^k-1)})$, which is a $(\beta,\epsilon)$-hopset for any $0<\epsilon<1$ with $\beta=O(k/\epsilon)^k$.
\end{theorem}

Since our construction is based on the \cite{TZ06} emulator construction, following their analysis we obtain an emulator with additive stretch that can have {\em linear} size.
\begin{corollary}\label{col:emul}
For any un-weighted graph $G=(V,E)$ on $n$ vertices and any $k\ge 1$, there exists an emulator $H$ of size at most $O(n^{1+1/(2^k-1)})$, with additive stretch $O(k\cdot d^{1-1/k})$ for pairs at distance $d$.
\end{corollary}

\subsection{Efficient Implementation}\label{sec:implem}

We consider the construction and notation of \sectionref{sec:smaller-size}, with the slightly stronger assumption $k\le \log\log n-2$.

It was (implicitly) shown in \cite{TZ01} that for any $0\le i< k$, the sets $B(u)$ (and the corresponding distances) for all $u\in A'_i\setminus A'_{i+1}$ can be computed in expected time $O(|E|+n\log n)/p'_i$. (In fact, in \cite{TZ01}, the set $B(u)$ contains more vertices, not only those in $A'_i$.  However, we can remove the extra vertices easily.) The running time becomes larger as $i$ grows, and in order to keep it under control, we use the method of \cite{EN16}: introduce a parameter $2 \nu <\rho<1$, and redefine the probabilities as follows.
Set $i_0=\lfloor\log(\rho/\nu)\rfloor$ and $i_1=i_0+1+\left\lceil \frac{1}{\rho}\right\rceil$. For $0\le i\le i_0$, let $p_i'=n^{-2^i\cdot\nu}\cdot 2^{2^i-1}$ as in Section \ref{sec:smaller-size}. Set also $p'_{i_0+1}=n^{-\rho/2}$, and for the remaining levels $i_0+2\le i\le i_1$, set $p_i'=n^{-\rho}$. Finally, let $A'_{i_1+1}=\emptyset$. Note that for $0\le i\le i_0+1$, we have that
\begin{equation}\label{eq:NI}
N'_i=n\prod_{j=0}^{i-1}p'_j=n^{1-(2^i-1)\nu}\cdot 2^{2^i-i-1}~.
\end{equation}
In particular, using that $\rho/\nu\le 2^{i_0+1}\le 2\rho/\nu$, we get
\begin{equation}\label{eq:i0+1}
N'_{i_0+1}\le n^{1-(\rho/\nu-1)\nu}\cdot 2^{2\rho/\nu}\le n^{1+\nu-\rho/2}~.
\end{equation}
(The last inequality uses that $k\le\log\log n-2$. Thus $1/\nu=(2^k-1)\le\log n/4$, and so that $2^{2\rho/\nu}\le n^{\rho/2}$.) Thus for any $i\ge i_0+2$ we see that
\begin{equation}\label{eq:i0+2}
N'_i= N'_{i_0+1}\prod_{j=i_0+1}^{i-1}p'_j\stackrel{\eqref{eq:i0+1}}{\le} n^{1+\nu-\rho/2}\cdot n^{-\rho/2}\cdot n^{-(i-1-(i_0+1))\cdot\rho}=n^{1+\nu}\cdot n^{-(i-(i_0+1))\cdot\rho}~.
\end{equation}

The expected number of edges inserted until phase $i_0+1$ is at most
\[
\sum_{i=0}^{i_0}N'_i/p'_i\stackrel{\eqref{eq:NI}}{\le}\sum_{i=0}^{i_0}n^{1-(2^i-1)\nu}\cdot 2^{2^i-i-1}\cdot n^{2^i\cdot\nu}/2^{2^i-1}=\sum_{i=0}^{i_0}n^{1+\nu}/2^i\le 2n^{1+\nu}~.
\]
The expected number of edges at phase $i_0+1$ is bounded by
\[
N'_{i_0+1}/ p'_{i_0+1}\stackrel{\eqref{eq:i0+1}}{\le} n^{1+\nu-\rho/2}\cdot n^{\rho/2}=n^{1+\nu}~.
\]
The remaining phases until $i_1$ introduce at most
\[
\sum_{i=i_0+2}^{i_1-1}N'_i/p'_i\stackrel{\eqref{eq:i0+2}}{\le} \sum_{i=i_0+2}^{i_1-1}n^{1+\nu}\cdot n^{-(i-(i_0+1))\cdot\rho}\cdot n^\rho \le n^{1+\nu}\cdot\sum_{i=0}^\infty n^{-i\rho}\le 2n^{1+\nu}~,
\]
as this summation converges. Finally, since
\begin{equation}\label{eq:last-bunch}
N'_{i_1}\le n^{1+\nu}\cdot n^{-(i_1-(i_0+1))\cdot\rho}\le n^\nu~,
\end{equation}
the last phase $i_1$ contributes at most $N'_{i_1}\cdot N'_{i_1}\le n^{2\nu}\le n^{1+\nu}$ edges to the hopset.
We conclude that the total number of edges is $O(n^{1+\nu})$.

Recall that the expected running time of the Dijkstra explorations at level $i<i_1$ is $O(|E|+n\log n)/p'_i$. Thus the expected running time of the first $i_0$ levels converges to $O(|E|+n\log n)\cdot n^\rho$, while each of the at most $\lceil 1/\rho\rceil+1$ remaining levels will take also $O(|E|+n\log n)\cdot n^\rho$ time. The final level is expected to take $O(|E|+n\log n)\cdot n^\nu$ as well, since there are expected $O(n^\nu)$ vertices at $A_{i_1}$ from which Dijkstra is performed.
The price we pay is in a higher number of sets, which increases the exponent of $\beta$ by at most an additive $1/\rho$. The following result summarizes the discussion.

\begin{theorem}
For any weighted graph $G=(V,E)$ on $n$ vertices, and any $1< k\le \log\log n - 2$, $0<\rho<1$, there is a randomized algorithm that runs in expected time $O(|E|+n\log n)\cdot n^\rho/\rho$, and computes an edge set $H$ of size at most $O(n^{1+1/(2^k-1)})$.  This edge set $H$  is a $(\beta,\epsilon)$-hopset for any $0<\epsilon<1$, where
\[
\beta=O\left(\frac{k+1/\rho}{\epsilon}\right)^{k+1/\rho+1}~.
\]
\end{theorem}

By substituting $\kappa=2^k-1$, we obtain an improved version of the hopsets of \cite{EN16}, where both the size of the hopset and the running time are smaller by a factor of $\log n$, while the other parameters remain the same. Another notable advantage is that it yields a single hopset which works for all $0<\epsilon<1$ simultaneously. 

\begin{corollary}\label{cor:tz}
For any weighted graph $G=(V,E)$ on $n$ vertices, and any $\kappa\ge 1$, $0<\rho<1$, there is a randomized algorithm that runs in expected time $O(|E|+n\log n)\cdot n^\rho/\rho$, and computes $H$ of size at most $O(n^{1+1/\kappa})$, which is a $(\beta,\epsilon)$-hopset for any $0<\epsilon<1$, where
\[
\beta=O\left(\frac{\log\kappa+1/\rho}{\epsilon}\right)^{\log\kappa+1/\rho+1}~.
\]
\end{corollary}

\section{Distributed Models}\label{sec:dist}

We will consider two standard models in distributed computing: the Congested Clique model, and the CONGEST model. In both models every vertex of an $n$-vertex graph $G = (V,E)$ hosts a processor, and the processors communicate with one another in discrete rounds, via short messages. Each message is allowed to contain an identity of a vertex, an edge weight, a distance in the graph, or anything else of no larger (up to a fixed constant factor) size.\footnote{Typically, in the CONGEST model only messages of size $O(\log n)$ bits are allowed, but edge weights are restricted to be at most polynomial in $n$. Our definition is geared to capture a more general situation, when there is no restriction on the aspect ratio. Hence results achieved in our more general model are more general than previous ones.} The local computation is assumed to require zero time, and we are interested in algorithms that run for as few rounds as possible.
In the Congested Clique model, we assume that all vertices are interconnected via direct edges, while in the CONGEST model, every vertex can send messages only to its $G$-neighbors (the weight of edges is irrelevant to the communication time).

\subsection{Congested Clique Model}\label{sec:CC}

We first show how to construct the hopset in the Congested Clique model. In order to avoid a high number of rounds when computing distances for determining the bunches $\{B(u)\}$, we built the hopset in $\log n$ levels, where each level $\ell$ hopset will only "take care" of pairs that have a shortest path with {\em at most $2^\ell$ hops}. This is somewhat different from previous works \cite{C00,N14,HKN16,EN16}, in which the level $\ell$ hopset handled pairs with {\em distance in the range $[2^\ell,2^{\ell+1}]$}. A few advantages of our current approach: it easily avoids the dependency on the ratio between largest and smallest distance, and also the final hopset is just the level $\log n$ one, so we can obtain a linear size hopset (unlike previous works which took the union of all levels). 

There are a few technical difficulties in implementing the algorithm of \sectionref{sec:hop} in a distributed setting. The first is that the \cite{TZ01} method for computing the bunches $B(u)$ was to compute their "inverses" -- called clusters.\footnote{The cluster $C(v)$ is defined as follows: each point $u\in C(v)$ iff $v\in B(u)$.} Alas, it is not known how to compute these clusters in a distributed manner when errors are allowed. Rather, we compute the bunches directly, and to avoid the potential large congestion (a vertex may be a part of many bunches, and needs to send messages for all of them), we replace the bunches with {\em half-bunches} (i.e., taking only points closer than half the distance to the pivot). See below for the formal definition. The second issue (of congestion) is more subtle, and arises since hop-bounded distances do not obey the triangle inequality. For the stretch analysis to go through, we need that the weight of the hopset edges will be bounded by a certain path between the end-points of the edge (see \eqref{eq:eq2}). In order to ensure that this happens, we build each hopset for level $\ell$ gradually, i.e. the bunches are created first for $A_0\setminus A_1$, then for $A_1\setminus A_2$ and so on, where each time the partial hopset is added to the graph on which we compute distance. 

We say that $H$ is a {\em $(\beta,\epsilon,t)$-hopset} if $G\cup H$ provides $(1+\epsilon)$ approximation with at most $\beta$ hops for all pairs $x,y\in V$ such that $d_G(x,y)=d_G^{(t)}(x,y)$ (i.e., the pairs that have a shortest path consisting of at most $t$ edges between them).
Note that the empty set is a $(1,0,1)$-hopset (and thus also a $(\beta,\epsilon,1)$-hopset for all $\beta\ge 1$ and $\epsilon>0$). Given a $(\beta,\epsilon_{\ell-1},2^{\ell-1})$-hopset $H^{(\ell-1)}$, we build a $(\beta,\epsilon_\ell,2^\ell)$-hopset $H^{(\ell)}$, where $1+\epsilon_\ell=(1+\epsilon)^\ell$ for some $0<\epsilon<1/(5\log n)$. The final hopset will be $H^{(\log n)}$. (We stress that the previous hopsets $H^{(1)},\dots,H^{(\ell-1)}$ are only used to compute $H^{(\ell)}$, and are not contained in it.)

Observe that $H^{(\ell-1)}$ is a $(2\beta,\epsilon_{\ell-1},2^\ell)$-hopset, since every path with at most $2^\ell$ hops can be partitioned into two paths of at most $2^{\ell-1}$ hops each, and $H^{(\ell-1)}$ provides a $1+\epsilon_{\ell-1}$ approximation with $\beta$ hops for each of these. It follows that for any $x,y\in V$,
\begin{equation}\label{eq:l-1}
d_{G\cup H^{(\ell-1)}}^{(2\beta)}(x,y)\le(1+\epsilon_{\ell-1})\cdot d_G^{(2^\ell)}(x,y)~.
\end{equation}

We sample sets $V=A_0\supseteq A_1\supseteq\dots\supseteq A_{k'+1}=\emptyset$ as in \sectionref{sec:implem}, where $k'=i_1\le k+1/\rho+1$ is the number of sets.  We introduce a subtle change to the construction -- in the previous section we defined for each $u\in A_i\setminus A_{i+1}$ a set $B(u)$, and added the edges $(u,v)$ for all $v\in B(u)$, simultaneously for all $0\le i\le k'$. Here we shall build the hopset $H=H^{(\ell)}$ gradually: For each $i=0,1,\ldots,k'$ we define a set of edges $H_i=H_i^{(\ell)}$ corresponding to the bunches of vertices in $V\setminus A_{i+1}$, and finally take $H=H_{k'}$.

Fix some $0\le i\le k'$, and assume we built the set $H_{i-1}$ (when $i=0$ define $H_{-1}=\emptyset$). We shall work in the graph $G_i$, defined by
\[
G_i=G\cup H^{(\ell-1)}\cup H_{i-1}~.
\]
The algorithm consists of two stages. On the first stage,
for $8\beta$ rounds, run a Bellman-Ford exploration in $G_i$ rooted at $A_{i+1}$, to obtain for each $u\in V$ the value $\hat{d}(u,A_{i+1})=d_{G_i}^{(8\beta)}(u,A_{i+1})$. (If some vertex $v\in V$ was not found in the exploration, then we set $\hat{d}(v,A_{i+1})=\infty$.)
Also for each vertex $u\in A_i\setminus A_{i+1}$ with $\hat{d}(u,A_{i+1})<\infty$, store $p(u)$ as a vertex $p\in A_{i+1}$ satisfying $\hat{d}(u,A_{i+1})=d_{G_i}^{(8\beta)}(u,p)$.

To determine the bunches (this is the second stage), each vertex $u\in A_{i}\setminus A_{i+1}$ conducts another Bellman-Ford exploration in the graph $G_i$ rooted at $u$, this time for only $4\beta$ rounds, i.e., half the number of hops $8\beta$ of the first exploration, to distance less than $\hat{d}(u,A_{i+1})/2$ (i.e., the messages whose origin is $u$ contain the value $\hat{d}(u,A_{i+1})$, and only vertices within this distance from $u$ forward the message in the next round). Define the half-bunch $B(u)=\{v\in A_{i}~:~ d_{G_i}^{(4\beta)}(u,v)<\hat{d}(u,A_{i+1})/2\}\cup\{p(u)\}$. Let
\begin{equation}
\label{eq:Hs}
H_i=H_{i-1}\cup\bigcup_{u\in A_{i}\setminus A_{i+1}}\{(u,v)~:~v\in B(u)\}~,
\end{equation}
and set the weight of the edge $(u,v)$ as the distance discovered in the exploration (i.e., $d_{G_i}^{(4\beta)}(u,v)$ for $v\in B(u)$ and $d_{G_i}^{(8\beta)}(u,p(u))$ for the pivot).

\begin{claim}\label{claim:size}
$\E[|H|]\le O(n^{1+\nu})$.
\end{claim}
\begin{proof}
The argument is essentially the same as the one in \sectionref{sec:implem}. The only difference is that when analyzing in step $0\le i<k'$ the expected size of a bunch, i.e., $\E[|B(u)|]$ for some $u\in A_i\setminus A_{i+1}$, we consider the ordering on $V$ given by the $4\beta$-bounded distance from $u$ in the graph $G_i$, i.e., according to $d_{G_i}^{(4\beta)}(u,\cdot)$. Then the size of $B(u)$ is bounded by the index of the first vertex in this ordering that is included in $A_{i+1}$. Since every $v\in A_i$ is included in $A_{i+1}$ independently with probability $p_i'$, we have that $\E[|B(u)|]\le 1/p_i'$. In fact, $B(u)$ may have a smaller size than the first index included in $A_{i+1}$, since we use more hops for computing distances to pivots (which reduces the distance threshold for being in $B(u)$), and since we only take into the bunch points that are less than half the distance to the pivot. Finally, for the last level $k'$ we have that for $u\in A_{k'}$, $\E[|B(u)|]\le\E[|A_{k'}|]\stackrel{\eqref{eq:last-bunch}}{\le} n^{\nu}$. Combining this with bounds on $N'_i=\E[|A_i|]$ in \sectionref{sec:implem} we can bound the size of $H$ in the same manner.
\end{proof}
In fact, since $|B(u)|$ is stochastically bounded by a geometric distribution with parameter $p'_i\ge n^{-\rho}$, it follows that with high probability for all $v\in V$,
\begin{equation}\label{eq:size-of-bunch}
|B(v)|\le 4n^\rho\cdot\ln n~.
\end{equation}

\begin{claim}\label{claim:mem}
The number of rounds required is whp $O(n^\rho\cdot k'\cdot\log^2n\cdot\beta)$.
\end{claim}
\begin{proof}
The sampling of the sets $A_i$ is done independently for each vertex, therefore it requires no communication. For each $1\le\ell\le\log n$ and $0\le i< k'$, we conduct a single Bellman-Ford exploration in $G_i$ rooted at $A_{i+1}$ for $8\beta$ rounds. Since in the Congested Clique model all edges are present, this requires $O(\beta)$ rounds per exploration (every vertex sends just a single message to all its neighbors every round). The more expensive step is the explorations to range $4\beta$ rooted at each $u\in A_i\setminus A_{i+1}$. The number of rounds in these explorations is affected by the number of messages a vertex $v\in V$ needs to forward to its neighbors at each round.
In what follows we prove that with high probability this number is at most $O(n^\rho\log n)$ for every $v\in V$.

Fix $0\le i<k'$. Consider $v\in V$, and order the vertices of $A_i$ according to their $4\beta$-bounded distance to $v$ in $G_i$, that is, according to $d_{G_i}^{(4\beta)}(v,\cdot)$. Since $p_i'\ge n^{-\rho}$, the probability that none of the first $2n^\rho\ln n$ vertices in that ordering is sampled to $A_{i+1}$ is at most
\[
\left(1-n^{-\rho}\right)^{2n^\rho\ln n}\le 1/n^2~.
\]
So by the union bound on the $n$ vertices, with high probability, for each $v\in V$ at least one of the first $2n^\rho\ln n$ vertices in its ordering of $A_i$ is sampled to $A_{i+1}$. Denote by $z\in A_{i+1}$ the first vertex in the ordering of $v$ that was chosen to $A_{i+1}$. We claim that no vertex $u\in A_i$, that appears after $z$ in the ordering of $v$, will cause $v$ to forward messages concerning $B(u)$. This is because in the first stage we performed the Bellman-Ford exploration rooted at $A_{i+1}$ for $8\beta$ rounds. Thus
\[
\hat{d}(u,A_{i+1})\le d_{G_i}^{(8\beta)}(u,z)\le d_{G_i}^{(4\beta)}(u,v)+d_{G_i}^{(4\beta)}(v,z)\le 2d_{G_i}^{(4\beta)}(u,v)~,
\]
where the last inequality uses the assumption that $u$ appeared after $z$ in $v$'s ordering. We obtained $d_{G_i}^{(4\beta)}(u,v)\ge\hat{d}(u,A_{i+1})/2$. So by the definition of half bunch, $v\notin B(u)$ and thus, $v$  will not forward $u$'s messages.

We still have to argue about the last level $i=k'$ (since no vertex is chosen to $A_{k'+1}$). Recall that the expected size of $A_{k'}$ is bounded by $n^{\nu}$, as shown in \eqref{eq:last-bunch}. It can be easily checked that whp $|A_{k'}|\le O(n^\nu\cdot\log n)\le O(n^\rho\cdot\log n)$.
(Recall that $\rho \ge 2\nu$.)
We conclude that whp, every vertex needs to send at most $O(n^\rho\log n)$ messages to implement a single step of Bellman-Ford. There are $O(\beta)$ rounds for each $\ell=1,\dots,\log n$ and each $0\le i\le k'$, so the total number of rounds required is $O(n^\rho\cdot k'\cdot\log^2n\cdot\beta)$.
\end{proof}

Next, we prove an analogue of \lemmaref{lem:main} for the distributed setting. There are several subtle differences described in the beginning of this section. So we provide a complete proof that addresses these subtleties.

\begin{lemma}\label{lem:dist-main}
Fix any $0<\delta<1/(15k')$, set $\beta=(3/\delta)^{k'}$, and let $x,y\in V$ be such that $d_G(x,y)=d_G^{(2^\ell)}(x,y)$. Then for every $0\le i\le k'$, at least one of the following two assertions holds:
\begin{enumerate}
\item $d_{G\cup H_i}^{((3/\delta)^i)}(x,y)\le (1+\epsilon_{\ell-1})\cdot(1+12\delta i)\cdot d_G(x,y)$.
\item There exists $z\in A_{i+1}$ such that $d_{G\cup H_i}^{((3/\delta)^i)}(x,z)\le 3 d_G(x,y)$.
\end{enumerate}
\end{lemma}
\begin{proof}
The proof is by induction on $i$. We start with the base case $i=0$. If it is the case that $x\in A_1$ then we can take $z=x$ and the second item holds trivially. Otherwise, consider the case that $x\in A_0\setminus A_1$ and $y\in B(x)$. Then we added an edge $(x,y)$ to $H_0$ of weight $d_{G_0}^{(4\beta)}(x,y)$ (the case that $y=p(x)$ is similar, replacing $4\beta$ by $8\beta$). Recall that $G_0=G\cup H^{(\ell-1)}$. Hence
\[
d_{H_0}^{(1)}(x,y)= d_{G_0}^{(4\beta)}(x,y) \le d_{G_0}^{(2\beta)}(x,y) \stackrel{\eqref{eq:l-1}}{\le} (1+\epsilon_{\ell-1})\cdot d_G^{(2^\ell)}(x,y) =(1+\epsilon_{\ell-1})\cdot d_G(x,y) ~,
\]
so the first item holds. The last case is that $x\in A_0\setminus A_1$ and $y\notin B(x)$. By definition of $H_0$, it must be that
\begin{equation}\label{eq:ynotin}
\hat{d}(x,A_1)\le 2d_{G_0}^{(4\beta)}(x,y)~.
\end{equation}
Since $p(x)\in B(x)$, the  edge $(x,p(x))$ is in the hopset, and its  weight is
\[
d^{(8\beta)}_{G_0}(x,p(x)) = \hat{d}(x,A_1)\stackrel{\eqref{eq:ynotin}}{\le} 2d_{G_0}^{(4\beta)}(x,y)\le 2(1+\epsilon_{\ell-1})\cdot d_G(x,y)<3d_G(x,y)~,
\]
where for the last two inequalities  we again used \eqref{eq:l-1} and the fact that $1+\epsilon_{\ell-1}< 3/2$. (The latter holds since we assume $\epsilon<1/(5\log n)$ and $\ell\le\log n$, so $(1+\epsilon_\ell)=(1+\epsilon)^\ell<e^{1/5}$). This proves the second item with $z=p(x)\in A_1$.

Assume the claim holds for $i$, and we prove for $i+1$. Consider the shortest path $\pi(x,y)$ between $x,y$ in $G$ that contains at most $2^\ell$ edges, and partition it into $J \le 1/\delta$ segments $\{L_j=[u_j,v_j]\}_{j\in[J]}$ as in the proof of \lemmaref{lem:main}.
We use the induction hypothesis for all pairs $(u_j,v_j)$ with parameter $i$. (By the virtue of lying on a shortest path that has at most $2^\ell$ edges, all these pairs satisfy $d_G^{(2^\ell)}(u_j,v_j)=d_G(u_j,v_j)$). Consider first the case that the first item holds for all of them, that is, $d_{G\cup H_i}^{((3/\delta)^i)}(u_j,v_j)\le (1+\epsilon_{\ell-1})\cdot(1+12\delta i)\cdot d_G(u_j,v_j)$. Then we take the path in $G\cup H_i$ that consists of the $(3/\delta)^i$-hops between each pair $u_j,v_j$, and the edges $(v_j,u_{j+1})$ of $G$. Since by (\ref{eq:Hs}),  $H_i\subseteq H_{i+1}$,  we have
\[
d_{G\cup H_{i+1}}^{((3/\delta)^{i+1})}(x,y)\le \sum_{j\in[J]}(d_{G\cup H_i}^{((3/\delta)^i)}(u_j,v_j)+d_G^{(1)}(v_j,u_{j+1}))\le(1+\epsilon_{\ell-1})\cdot(1+12\delta i)\cdot d_G(x,y)~,
\]
which concludes the proof for the first case.
The second case is that there are pairs $(u_j,v_j)$ for which only the second item holds. Let $l\in [J]$ (resp., $r\in[J]$) be the first (resp., last) index for which the first item does not hold for the pair $(u_l,v_l)$ (resp., $(u_r,v_r)$). Then there are $z_l,z_r\in A_{i+1}$ such that
\begin{equation}\label{eq:ggf}
d_{G\cup H_i}^{((3/\delta)^i)}(u_l,z_l)\le 3d_G(u_l,v_l)~\textrm{ and }~ d_{G\cup H_i}^{((3/\delta)^i)}(v_r,z_r)\le 3d_G(u_r,v_r)~.
\end{equation}
Consider first the case that $z_l\in A_{i+2}$. Then we take $z=z_l$, and derive
\begin{eqnarray*}
d_{G\cup H_{i+1}}^{((3/\delta)^{i+1})}(x,z)&\le& \sum_{j=1}^{l-1}\left(d_{G\cup H_i}^{((3/\delta)^i)}(u_j,v_j)+d_G^{(1)}(v_j,u_{j+1})\right)+d_{G\cup H_i}^{((3/\delta)^i)}(u_l,z_l)\\
&\stackrel{\eqref{eq:ggf}}{\le}&(1+\epsilon_{\ell-1})\cdot(1+12\delta i)\cdot d_G(x,u_l) + 3d_G(u_l,v_l)\\
&\le& 3d_G(x,y)~,
\end{eqnarray*}
where in the second inequality we used that the first item holds for all intervals until the $l$-th one, and in the final one that $1+\epsilon_{\ell-1}<3/2$ and $1+12\delta i<2$.

From now on assume $z_l\in A_{i+1}\setminus A_{i+2}$. Recall that the Bellman-Ford explorations that constructed  $H_{i+1}$ were conducted in the graph $G_{i+1}=G\cup H^{(\ell-1)}\cup H_i$. These explorations were conducted to hop-depth $8\beta$ on the first stage, and $4\beta$ on the second. This allows us to provide the following bound:
\begin{eqnarray}\label{eq:eq2}
d_{G_{i+1}}^{(4\beta)}(z_l,z_r)
&\le&d_{G\cup H_i}^{(\beta)}(z_l,u_l)+d_{G\cup H^{(\ell-1)}}^{(2\beta)}(u_l,v_r)+d_{G\cup H_i}^{(\beta)}(v_r,z_r)\\\nonumber
&\stackrel{\eqref{eq:l-1}}{\le}&d_{G\cup H_i}^{((3/\delta)^i)}(z_l,u_l)+(1+\epsilon_{\ell-1})\cdot d_G(u_l,v_r)+d_{G\cup H_i}^{((3/\delta)^i)}(v_r,z_r)~.
\end{eqnarray}
Here the first inequality follows by the triangle inequality, the second uses that $(3/\delta)^i\le\beta$,  that $u_l,v_r$ lie on a shortest path with at most $2^\ell$ hops, and that $H^{(\ell-1)}$ is a $(2\beta,\eps_{\ell-1},2^\ell)$ hopset.

Consider the case that $z_r\in B(z_l)$, then we have a hopset edge $(z_l,z_r)$ that was introduced in $H_{i+1}$. In particular, since we used $4\beta$ steps in the exploration from $z_l$, we have that
\begin{equation}\label{eq:eq22}
d_{H_{i+1}}^{(1)}(z_l,z_r)= d_{G_{i+1}}^{(4\beta)}(z_l,z_r)\stackrel{\eqref{eq:eq2}}{\le}d_{G\cup H_i}^{((3/\delta)^i)}(z_l,u_l)+(1+\epsilon_{\ell-1})\cdot d_G(u_l,v_r)+d_{G\cup H_i}^{((3/\delta)^i)}(v_r,z_r)~.
\end{equation}

Next, apply the inductive hypothesis on segments $\{L_j\}$ for $j<l$ and $j>r$, and in between use the detour via $u_l,z_l,z_r,v_r$. Since there are at most $1/\delta-1$ intervals for which we use the first item in the inductive hypothesis, the total number of hops we will need is at most $(1/\delta-1)\cdot(3/\delta)^i+1/\delta+2(3/\delta)^i+1$. This is at most $(3/\delta)^{i+1}$ whenever $\delta<1/2$. It follows that
\begin{eqnarray*}
d_{G\cup H_{i+1}}^{((3/\delta)^{i+1})}(x,y)
&\le& \sum_{j=1}^{l-1}\left[d_{G\cup H_i}^{((3/\delta)^i)}(u_j,v_j)+d_G^{(1)}(v_j,u_{j+1})\right]+d^{((3/\delta)^i)}_{G\cup H_i}(u_l,z_l)+ d_{H_{i+1}}^{(1)}(z_l,z_r)\\
&&~+ d^{((3/\delta)^i)}_{G\cup H_i}(z_r,v_r)+d_G^{(1)}(v_r,u_{r+1})+ \sum_{j=r+1}^{J}\left[d_{G\cup H_i}^{((3/\delta)^i)}(u_j,v_j)+d_G^{(1)}(v_j,u_{j+1})\right]\\
&\stackrel{\eqref{eq:eq22}}{\le}& (1+\epsilon_{\ell-1})\cdot\Big[(1+12\delta i)\cdot d_G(x,u_l)+d_G(u_l,v_r)+(1+12\delta i)\cdot d_G(v_r,y)\Big]\\
&&~+2d^{((3/\delta)^i)}_{G\cup H_i}(u_l,z_l)+2d^{((3/\delta)^i)}_{G\cup H_i}(z_r,v_r)\\
&\stackrel{\eqref{eq:ggf}}{\le}&(1+\epsilon_{\ell-1})\cdot(1+12\delta i)\cdot d_G(x,y)+12\delta\cdot d_G(x,y)\\
&\le&(1+\epsilon_{\ell-1})\cdot(1+12\delta(i+1))\cdot d_G(x,y)~.
\end{eqnarray*}
In the penultimate inequality we used that both $d_G(u_l,v_l),d_G(u_r,v_r)\le \delta\cdot d_G(x,y)$.
This demonstrates that item 1 holds in this case.

The final case to consider is that $z_r\notin B(z_l)$ (and $z_l\in A_{i+1}\setminus A_{i+2}$). Let $z=p(z_l)\in A_{i+2}$. Since $z_r\in A_{i+1}$, the definition of $B(z_l)$ implies that
\begin{equation}\label{eq:frrf}
d^{(1)}_{H_{i+1}}(z_l,z)=\hat{d}(z_l,A_{i+2})\le 2d_{G_{i+1}}^{(4\beta)}(z_l,z_r)~.
\end{equation}
(Recall that $G_{i+1} = G \cup H^{(\ell-1)} \cup H_i$.)

We now claim that item 2 holds for such a choice of $z$. Indeed,
by (\ref{eq:ggf}), we have
\begin{equation}
\label{eq:my_ggf}
3 \cdot d_{G \cup H_i}^{(3/\delta)^i}(u_l,z_l) + 2 \cdot d_{G \cup H_i}^{(3/\delta)^i}(v_r,z_r) ~\le~ 15 \cdot d_G(x,y)~.
\end{equation}
Hence,
\begin{eqnarray*}
d_{G\cup H_{i+1}}^{((3/\delta)^{i+1})}(x,z)&\le&\sum_{j=1}^{l-1}\left[d_{G\cup H_i}^{((3/\delta)^i)}(u_j,v_j)+d_G^{(1)}(v_j,u_{j+1})\right]+d^{((3/\delta)^i)}_{G\cup H_i}(u_l,z_l)+d^{(1)}_{H_{i+1}}(z_l,z)\\
&\stackrel{\eqref{eq:frrf}\wedge\eqref{eq:eq2}}{\le}&(1+\epsilon_{\ell-1})\cdot(1+12\delta i)\cdot d_G(x,u_l)+d^{((3/\delta)^i)}_{G\cup H_i}(u_l,z_l) \\
&&~+ 2\cdot\left[d_{G\cup H_i}^{((3/\delta)^i)}(z_l,u_l)+(1+\epsilon_{\ell-1})\cdot d_G(u_l,v_r)+d_{G\cup H_i}^{((3/\delta)^i)}(v_r,z_r)\right]\\
&\stackrel{\eqref{eq:my_ggf}}{\le}&(1+\epsilon_{\ell-1})\cdot(1+12\delta i)\cdot d_G(x,u_l) +15\delta\cdot d_G(x,y) +2(1+\epsilon_{\ell-1})\cdot d_G(u_l,v_r)\\
&\le& 3d_G(x,y)~,
\end{eqnarray*}
where the last inequality we used that $\delta<1/(15k)$, $k\ge 2$ and $1+\epsilon_{\ell-1}<e^{1/5}<5/4$, so that both $(1+\epsilon_{\ell-1})\cdot(1+12\delta i)+15\delta\le 3$ and $2(1+\epsilon_{\ell-1})+15\delta\le 3$.
\end{proof}

Taking $\delta=\epsilon/(15k')$ and picking $i=k'$, the second item of \lemmaref{lem:dist-main} cannot hold for any $x,y\in V$ (because $A_{k'+1}=\emptyset$), so we have for every $x,y\in V$ such that $d_G^{(2^\ell)}(x,y)=d_G(x,y)$ that
\[
d_{G\cup H^{(\ell)}}^{(\beta)}(x,y)\le (1+\epsilon_{\ell-1})\cdot(1+\epsilon)\cdot d_G^{(2^\ell)}(x,y)=(1+\epsilon_\ell)\cdot d_G(x,y)~.
\]
Recall that $\beta = (3/\delta)^{k'} = (45k'/\eps)^{k'}$.
Rescaling $\epsilon'=\epsilon/\log n$ and taking $\ell=\log n$, we derive the following theorem.

\begin{theorem}\label{thm:dist-main}
For any weighted graph $G=(V,E)$ on $n$ vertices, an integer $k> 1$, and parameters $0<\rho<1$, $0<\epsilon<1/5$, there is a distributed algorithm in the Congested Clique model running in $\tilde{O}(n^\rho\cdot\beta)$ rounds, that computes a $(\beta,\epsilon)$-hopset $H$ of size at most $O(n^{1+1/(2^k-1)})$, where
\begin{equation}\label{eq:beta}
\beta=O\left(\frac{(k+1/\rho)\cdot\log n}{\epsilon}\right)^{k+1/\rho+1}~.
\end{equation}
\end{theorem}
\begin{remark}
We note that by \eqref{eq:size-of-bunch} and \claimref{claim:mem}, the memory requirement from every vertex is $\tilde{O}(n^\rho)$. This is because the latter shows that this is a bound on the number of messages every vertex needs to send in each round, and the former indicates that whp storing $B(v)$ for any $v\in V$ requires only so much space.
\end{remark}



We remark that one can achieve $\beta$ independent of $n$ by either applying the construction recursively, as we do in \sectionref{sec:pram} for the parallel implementation, or by using an idea from \cite{EN16}. We next describe the latter: fix a parameter $t$, and use the hopset $H^{(\ell)}$ to compute the hopset $H^{(\ell+t)}$; Since $H^{(\ell)}$ is also a $(2^t\cdot\beta,\epsilon,2^{\ell+t})$-hopset, we need explorations to range $2^t\cdot\beta$ in order for an appropriate variant of \eqref{eq:l-1} to hold. There will be only $(\log n)/t$ levels until $H^{(\log n)}$ is built, so we gain a factor of $t$ in $\beta$. We derive the following result.
\begin{theorem}
For any weighted graph $G=(V,E)$ on $n$ vertices, integers $k> 1$, $t\ge 1$, and parameters $0<\rho<1$, $0<\epsilon<1/5$, there is a distributed algorithm in the Congested Clique model that runs in $\tilde{O}(n^\rho\cdot\beta\cdot 2^t/t)$ rounds, and computes $H$ of size at most $O(n^{1+1/(2^k-1)})$, which is a $(\beta,\epsilon)$-hopset, where
\[
\beta=O\left(\frac{(k+1/\rho)\cdot\log n}{t\cdot\epsilon}\right)^{k+1/\rho}~.
\]
\end{theorem}

In particular, taking $t=\rho\log n$ and rescaling $\rho'=2\rho$, gives
\begin{corollary}
For any weighted graph $G=(V,E)$ on $n$ vertices, an integer $k> 1$, and parameters $0<\rho<1/2$, $0<\epsilon<1/5$, there is a distributed algorithm in the Congested Clique model that runs in $\tilde{O}(n^\rho\cdot\beta)$ rounds, and computes $H$ of size at most $O(n^{1+1/(2^k-1)})$, which is a $(\beta,\epsilon)$-hopset, where
\[
\beta=O\left(\frac{k+1/\rho}{\rho\cdot\epsilon}\right)^{k+2/\rho}~.
\]
\end{corollary}

\subsection{CONGEST Model}\label{sec:CONGEST}

Given a weighted graph $G=(V,E,w)$ representing the network, in the CONGEST model we will be interested in a setting where there is a "virtual" graph $G'=(V',E',w')$ embedded in $G$, i.e., $V'\subseteq V$. We would like to construct a hopset for $G'$. It is motivated by distributed applications of hopsets for approximate shortest paths computation, distance estimation and routing \cite{HKN14,N14,HKN16,LP15,EN16a,EN16}, which require a hopset for a virtual graph embedded in the underlying network in the above way.

In a similar manner to \cite{EN16}, we can modify our algorithm in the Congested Clique model to the CONGEST model.
The following lemma provides a way to perform Bellman-Ford exploration using small memory.

\begin{lemma}\label{lem:BF-small}
Let $G''=(V',E'\cup H)$ be a virtual graph on $m$ vertices embedded in a graph $G=(V,E)$ of hop-diameter $D$, such that edges in $E'$ correspond to $B$-bounded distances in $G$, and $H$ has arboricity $\alpha$ (i.e., one can orient the edges of $H$ to have out-degree at most $\alpha$).
Moreover, every vertex $v' \in V'$ knows at most $\alpha$ its outgoing edges in $H$.
 Then one can compute $\beta$ iterations of Bellman-Ford in $G''$ in the CONGEST model within $O(m\cdot \alpha+B+D)\cdot\beta\cdot\log n$ rounds, so that every vertex requires only $O(\alpha\log n)$ memory. 
\end{lemma}
\begin{proof}
To implement a single iteration of the Bellman-Ford exploration, every vertex $v\in V'$, which holds a current distance estimate, will need to communicate it to its neighbors in $G''$. First, it will initiate an exploration in $G$ for $B$ rounds. In each round, every vertex $u\in V$ will forward the smallest value it received so far. This guarantees that if $\{v,w\}\in E'$, then $w$ will receive $v$'s message (or a smaller value). 

We now have to handle the edges of $H$.
Let $T$ be a spanning tree of $G$ with hop-depth $D$. Every $v\in V'$ will broadcast via $T$ its value to the entire graph, and will also send all the existing edges of $H$ incident on it that $v$ knows about. All vertices $w\in V'$ that know of a hopset edge $\{v,w\}$ (or that learn about it from $v$'s message) will update their value accordingly. Since there are $O(m\cdot \alpha)$ messages, this can be done in $O(m\cdot \alpha+D)$ rounds. In order to guarantee small internal memory, each $v$ selects at random a number from $\{1,2,\dots,m\cdot \alpha\}$ for each message it sends, as a round to start its broadcast (clearly this increases the number of rounds by at most $m\cdot \alpha$). Since
each message of $v$ will reach every vertex of $T$ at most once, the probability that some $u\in V$ receives $t$ messages in a single round is at most ${m\cdot\alpha\choose t}\cdot 1/(m\cdot \alpha)^t\le (e/t)^t$.  Thus, with high probability, no vertex will receive more than $O(\log n)$ messages each round. By increasing the number of rounds by $O(\log n)$, whp there will be no congestion. The total number of rounds required is thus $O(m\cdot \alpha+B+D)\cdot\beta\cdot\log n$.
\end{proof}

We now show how to use \lemmaref{lem:BF-small} to construct a hopset for $G'$, in the setting where $E'$ are edges corresponding to $B=\tilde{O}(m)$-bounded distances in $G$ (without computing $G'$ explicitly).
Recall that in the $i$th iteration of constructing $H=H^{(\ell)}$, we have already built the previous hopset $H^{(\ell-1)}$ and the partial hopset $H_{i-1}$. Since we desire limited memory, every vertex $v$ stores only the "outgoing" hopset edges, those to vertices in its bunch $B(v)$. Recall that by \eqref{eq:size-of-bunch}, whp $|B(v)|\le O(m^\rho\cdot\log n)$, for all $v\in V'$.

We work in the graph $G_i=G'\cup H^{(\ell-1)}\cup H_{i-1}$. In order to implement the $O(\beta)$-bounded exploration rooted at $A_{i+1}$
 (the second stage of the $i$th iteration), we simply apply \lemmaref{lem:BF-small} on $G_i$ with $\alpha=O(m^\rho\cdot\log n)$.
The explorations from vertices of $A_i\setminus A_{i+1}$ (the second stage of the $i$th iteration) are done in a similar manner. However, there is a  larger congestion than in the first stage, due to the multiple sources of limited explorations. Recall that in the limited exploration whose origin is $v\in A_i\setminus A_{i+1}$, each intermediate node $x\in V'$ forwards the message iff its current estimate is strictly less than $\hat{d}(v,A_{i+1})/2$ (this value is part of the message $v$ sends). We enforce the exact same rule for vertices $u\in V$ as well. If a message concerning $v$ should pass in $G'$ from $x$ to its neighbor $y$, then all vertices on the $B$-bounded path in $G$ that implements the edge $(x,y)\in E'$ will have estimates smaller than that of $y$, therefore will forward the message on.
In the proof of \claimref{claim:mem} we saw that each $x\in V'$ participates whp in at most $O(m^{\rho}\cdot \log n)$ explorations for each iteration $i$. The argument is identical for $u\in V$ as well, so the congestion induced in the first stage of \lemmaref{lem:BF-small} (the exploration in $G$ for $B$ rounds)
by multiple sources is only $O(m^{\rho}\cdot \log n)$. Note that in the second phase (broadcasting the edges of $H$), the number of messages increases to $O(m\cdot\alpha+m\cdot m^{\rho}\cdot \log n)$.  Thus, the total number of rounds required is still $\tilde{O}(m^{1+\rho}+D)\cdot\beta$.
We summarize the discussion with the following result.

\begin{theorem}\label{thm:dist-cong}
For any weighted graph $G=(V,E)$ with hop-diameter $D$, an integer $k> 1$, and parameters $0<\rho<1$, $0<\epsilon<1/5$, and (an implicit) virtual graph $G'=(V',E')$ embedded in $G$ on $|V'|=m$ vertices, there is a distributed algorithm in the CONGEST model that runs in $\tilde{O}(m^{1+\rho}+D)\cdot\beta$ rounds, that computes $H$, which is a $(\beta,\epsilon)$-hopset for $G'$, of size at most $O(m^{1+1/(2^k-1)})$, where
\[
\beta=O\left(\frac{(k+1/\rho)\cdot\log m}{\epsilon}\right)^{k+1/\rho+1}~.
\]
\end{theorem}
\begin{remark}\label{rem:2}
In the case that $E'$ corresponds to $B=\tilde{O}(m)$-bounded distances in $G$, the hopset can be computed where every vertex has internal memory $\tilde{O}(m^\rho)$.
\end{remark}

\paragraph{Path-reporting Hopsets:}

Every hopset edge is implemented via some path in $G$. For our application to routing, we would like that every vertex on a path implementing a certain hopset edge will be aware of this hopset edge. This means that for every hopset edge $(x,y)\in H$, there exists a path $P$ in $G$ of length $w_H(x,y)$, and every vertex $u\in P$ knows about the hopset edge, and the distances $d_P(u,x)$, $d_P(u,y)$, and its neighbors on $P$.
It was shown in \cite{EN16} how to adapt the Bellman-Ford exploration, so that paths information can be stored as well, at a cost of increasing the size of messages by a factor of $O(\beta)$. However, there was no guarantee on the number of hopset edges a vertex $u\in V$ can be a part of, which can be devastating when one desires small memory per vertex. We describe now an approach that eliminates the need for the message's size increase, and also ensures that each vertex belongs to a bounded number of paths that implement hopset edges. The issue that may cause a vertex to be in a path for many hopset edges, is that we use previous hopsets to construct a new one. Then the vertices implementing paths in these previous hopsets may not be discovered by the current explorations. So the argument of \claimref{claim:mem} bounding the number of explorations that visit a certain vertex does not apply as is.


In order to guarantee that every $u\in V$ will need to store information for only $\tilde{O}(m^\rho)$ hopset edges, we need to slightly change the construction. First, we will define $H^{(\ell)}=H_{k'}\cup H^{(\ell-1)}$, so that every hopset will contain all the previous hopsets.
(Recall that in our algorithm in Section \ref{sec:CC} that computes non-path-reporting hopsets,
 we only used lower-scale hopsets to compute a higher-scale one. Once the mission of lower-scale hopsets was completed, they were ruthlessly erased.)
Second, rather than performing the exploration from $A_{i+1}$ in $8\beta$ steps, we apply \lemmaref{lem:BF-small} with $8\beta\cdot k'\cdot\log n+1$ steps of Bellman-Ford. Note that in the proof of correctness we only used that there are at least $8\beta$ steps. Using more steps will only increase the number of rounds (by a poly-logarithmic factor). Recall that when computing the hopset $H^{(\ell)}$ at phase $i$, we have already computed $H_{i-1}$, and work in the graph $G_i=G'\cup H^{(\ell-1)}\cup H_{i-1}$. We can now argue that whp, there will not be too many hopset edges whose path in $G$ contains $u$. The intuition is that the exploration from $A_{i+1}$ has sufficiently many hops in order to discover this $u$, and so an argument similar to the one of \claimref{claim:mem} will apply.

Fix $u\in V$, and order the vertices of $A_i$ in increasing order according to their distance to $u$, where the distance from $v\in A_i$ to $u$ is the shortest path consisting of at most $4\beta\cdot k'\cdot\log n$ edges of $G_i$ and then at most $B$ edges of $G$. Let $z$ be the first vertex in that order that is included in $A_{i+1}$. We claim that the vertex  $u$ cannot belong to a path $P$ that implements a hopset edge $(x,y)$, such that $x\in A_i\setminus A_{i+1}$ is after $z$ in the ordering of $u$.

Consider how the path $P$ is built. One can initially start with $Q=\{(x,y)\}$, and then recursively replace the hopset edge in $Q$ that contains $u$, with the $4\beta$-bounded path in some $G'_j$ that induces it. Note that this recursion depth is at most $k'\log n$, thus $Q$ has at most $4\beta\cdot k'\log n$ edges. Since $G_i$ contains all the edges of all previous hopsets, the exploration from $A_{i+1}$ starting at $z$ for $8\beta\cdot k'\cdot\log n+1$ steps would have reached $u$ after $4\beta\cdot k'\cdot\log n+1$ edges of $G_i$ (the $B$ edges of $G$ are an edge of $G'$, and thus of $G_i$ as well), and then after additional $4\beta\cdot k'\cdot\log n$ edges of $G_i$, it would have surely have reached $y$ (because $z$ is closer to $u$ than $x$). We conclude that
\[
\hat{d}(y,A_{i+1})\le d_{G_i}^{(8\beta\cdot k'\cdot\log n+1)}(y,z)\le d_P(y,u)+\hat{d}(u,A_{i+1})\le d_P(y,x)+d_P(u,x)\le 2d_{G_i}^{(4\beta)}(x,y)~,
\]
which is a contradiction to the fact that $y$ joins $B(x)$.

Next, we have to show that each $u$ will indeed learn the relevant information on all hopset edges it implements. Assume inductively that for any hopset edge $(x,y)\in H^{(\ell-1)}\cup H_{i-1}$, if $P$ is the path in $G$ that implements this edge, then every $u\in P$ knows about the edge, $d_P(u,x)$, $d_P(u,y)$, and its neighbors on $P$. A new hopset edge $(x,y)\in H_i$ is created whenever the exploration rooted at $x\in A_i\setminus A_{i+1}$ discovers a vertex $y\in A_i$. Recall that this exploration is done in $G_i$ for $4\beta$ rounds. Whenever $y$ joins $B(x)$ it will send an acknowledgement on the $4\beta$-bounded path back to $x$ in $G_i$ (every vertex discovered by $x$ takes note of its "parent", the vertex who sent it the message of $x$). The acknowledgement phase can take place after the exploration concludes, and it will induce congestion that is no larger that the congestion created when sending the messages, so the number of rounds will at most double. Now, every vertex $v$ in the $4\beta$-bounded path from $y$ to $x$ that receives $y$'s acknowledgement, knows that the edge to its parent $v'$ is part of the path implementing the hopset edge $(x,y)$. Recall that the edge $(v,v')$ is either an edge of $G'$, which is discovered via a $B$-round exploration in $G$ -- in which case all vertices along the path in $G$ from $v$ to $v'$ can update the relevant information about $(x,y)$ when $v$ does a $B$-round exploration in $G$ (this is the acknowledgement step), or otherwise $(v,v')\in H^{(\ell-1)}\cup H_{i-1}$. In the latter case, $v$ will broadcast that the edge $(v,v')$ implements $(x,y)$, and its distances to $x,y$. By the induction hypothesis, each vertex $u'$ that implements a path $P'$ for the hopset edge $(v,v')$ knows about it and its distances to $v,v'$, thus when $u'$ hears this broadcast (which is sent to all vertices of $V$), it knows it implements $P$, and can computes distances to $x$ and $y$.

We conclude that whp every vertex needs to store only the $\tilde{O}(m^\rho)$ hopset edges that it implements.
Note that the final hopset $H^{(\log n)}$ can omit all the previous hopsets (which were used only for calculations). We summarize this discussion with the following theorem.
\begin{theorem}\label{thm:dist-cong-path}
For any weighted graph $G=(V,E)$ with hop-diameter $D$, an integer $k> 1$, and parameters $0<\rho<1$, $0<\epsilon<1/5$, and (an implicit)
 virtual graph $G'=(V',E')$ embedded in $G$ on $|V'|=m$ vertices, there is a distributed algorithm in the CONGEST model that runs in $\tilde{O}(m^{1+\rho}+D)\cdot\beta$ rounds, that computes $H$, which is a $(\beta,\epsilon)$ {\em path-reporting} hopset for $G'$, of size at most $O(m^{1+1/(2^k-1)})$, where
\[
\beta=O\left(\frac{(k+1/\rho)\cdot\log m}{\epsilon}\right)^{k+1/\rho+1}~.
\]
In the case that $E'$ corresponds to $B=\tilde{O}(m)$-bounded distances in $G$, the hopset can be computed where every vertex has internal memory $\tilde{O}(m^\rho)$.
\end{theorem}

\section{PRAM Model}\label{sec:pram}

The algorithm described in \sectionref{sec:CC} can be easily adapted to the PRAM model. For each $\ell=1,2,\dots,\log n$, we build the hopset $H^{(\ell)}$ based on the previous hopset $H^{(\ell-1)}$. Each of the $O(\beta)$-bounded Bellman-Ford explorations for constructing $H_i$ can be implemented in parallel in $O(\beta)$ rounds, where the congestion of $\tilde{O}(n^{\rho})$ per vertex translates to extra work (rather than multiplying the number of rounds, as was the case in  distributed models). Since there are $\log n$ values of $\ell$, and $k'\le k+1/\rho+1$ steps in each level, the number of rounds is only $O((k+1/\rho)\cdot\log n\cdot\beta)$. We have the following result.
\begin{theorem}\label{thm:pram-main}
For any weighted graph $G=(V,E)$ on $n$ vertices, an integer $k> 1$, and parameters $0<\rho<1$, $0<\epsilon<1/5$, there is a parallel algorithm running in $O((k+1/\rho)\cdot\log n\cdot\beta)$ rounds and has $\tilde{O}(|E|\cdot n^\rho)$ work, that computes $H$ of size at most $O(n^{1+1/(2^k-1)})$, which is a $(\beta,\epsilon)$-hopset, where
\begin{equation}\label{eq:beta1}
\beta=O\left(\frac{(k+1/\rho)\cdot\log n}{\epsilon}\right)^{k+1/\rho+1}~.
\end{equation}
\end{theorem}
We can also apply the construction recursively: If $H(1)$ is the hopset given by \theoremref{thm:pram-main} with $\beta_1 =\beta$ given in \eqref{eq:beta1}, then apply the construction on the graph $G\cup H(1)$, but only for levels $\ell$ up to $\ell_2=\log\beta_1$, to obtain a hopset $H(2)$. Since for any $x,y\in V$ we have $d_{G\cup H(1)}^{(\beta_1)}(x,y)\le (1+\epsilon)d_G(x,y)$, then adding both $H(1)$ and $H(2)$ guarantees $d_{G\cup H(1)\cup H(2)}^{(\beta_2)}(x,y)\le (1+\epsilon)^2d_G(x,y)$, where $\beta_2=\left(\frac{3c\cdot(k+1/\rho)\cdot\log \beta_1}{\epsilon}\right)^{k+1/\rho+1}$, where $c$ is the constant hidden by the $O(\cdot)$ notation in \eqref{eq:beta1}. This bound follows because $\epsilon$ needs to be rescaled by $3\ell_2=3\log\beta_1$; the rescaling by $\log\beta_1$ is to compensate for the number of levels, and by 3 to reduce the error from $(1+\epsilon)^2$ back to $1+\epsilon$. Continuing in this manner for the next level with $\ell_3=\log\beta_2$ levels, we obtain in general a recursion for $\beta_{i+1}=\left(\frac{3c\cdot(k+1/\rho)\cdot\log \beta_i}{\epsilon}\right)^{k+1/\rho+1}$, and it can be shown by induction that as long as $log^{(i)}n\ge 3c\log(k+1/\rho)$ we have
\[
\beta_i\le \left(\frac{8c\cdot(k+1/\rho)^2\cdot\left[\log(3c(k+1/\rho)/\epsilon)+\log^{(i)}n\right]}{\epsilon}\right)^{k+1/\rho+1}~.
\]
After at most $t=\log^*n$ iterations, we get that $\beta_t=O\left(\frac{(k+1/\rho)^2}{\epsilon}\right)^{(1+o(1))\cdot(k+1/\rho)}$. To summarize, this yields a hopset with constant parameter $\beta$ that is computed in ${\rm polylog}(n)$ rounds.

\begin{theorem}
For any weighted graph $G=(V,E)$ on $n$ vertices, an integer $k> 1$, and parameters $0<\rho<1$, $0<\epsilon<1/5$, there is a parallel algorithm running in $O(\left(\frac{(k+1/\rho)\cdot\log n}{\epsilon}\right)^{k+1/\rho+2})$ rounds and has $\tilde{O}(|E|\cdot n^\rho)$ work, that computes $H$ of size at most $O(n^{1+1/(2^k-1)}\cdot \log^*n)$, which is a $(\beta,\epsilon)$-hopset, where
\begin{equation}
\beta=O\left(\frac{(k+1/\rho)^2}{\epsilon}\right)^{(1+o(1))\cdot(k+1/\rho)}~.
\end{equation}
\end{theorem}

\section{Distributed Routing with Small Memory}\label{sec:route}

Here we improve the results of \cite{EN16a,LPP16}, and devise a compact routing scheme that can be efficiently implemented in a distributed network. The previous result of \cite{EN16a} provides, for any parameter $k$, a scheme with stretch $4k-5+o(1)$, labels of size $O(k\log^2n)$ and routing tables of size $O(n^{1/k}\log^2n)$.  The computation time of this scheme is  $(n^{1/2+1/k}+D)\cdot \min\{(\log n)^{O(k)},2^{\tilde{O}(\sqrt{\log n})}\}$ rounds (in the CONGEST model). One drawback of this result (and also of \cite{LPP16}, which obtained slightly weaker results), is that although the final memory requirement from each vertex is $\tilde{O}(n^{1/k})$, the preprocessing step requires high memory (at least $\Omega(\sqrt{n})$). Indeed, some of the classical works on compact routing schemes \cite{ABLP90} addressed the issue of each vertex having only a limited memory throughout the construction of the routing scheme (albeit their round complexity was at least linear in $n$). Here we present a distributed construction that has that desirable property, and in addition we improve both the label and table size by a logarithmic factor, almost matching the best known bounds of \cite{TZ01,C13} that are computed in a sequential manner.

We briefly sketch the approach of \cite{EN16a}, and the current improvement allowing low memory and improved bounds. First, construct the Thorup-Zwick hierarchy $V=A_0\supseteq A_1\supseteq A_k=\emptyset$, where each vertex in $A_{i-1}$ is sampled to $A_i$ independently with probability $n^{-1/k}$. Then the cluster $C(v)=\{u\in V~:~d_G(u,v)< d_G(u,A_{i+1})\}$ for $v\in A_i\setminus A_{i+1}$ can be viewed as tree rooted at $v$. Computing this cluster is done by a limited Dijkstra exploration from $v$, i.e., only vertices in $C(v)$ continue the exploration of $v$. Routing from $x$ to $y$ is done by finding an appropriate cluster $C(v)$ containing both $x,y$, and routing in that tree. Whenever $i<k/2$, these trees have whp depth $\tilde{O}(\sqrt{n})$. Hence they can be easily computed in a distributed manner within $\tilde{O}(n^{1/2+1/k})$ rounds. The main issue is computing the clusters for $i\ge k/2$.

The method of \cite{EN16a} was to work with a virtual graph $G'$, whose vertices are $V'=A_{k/2}$, and whose edges correspond to $B=c\cdot\sqrt{n}\log n$-bounded distances in $G$ between the vertices of $V'$. Then a hopset is computed for this virtual graph, which enables the computation of Bellman-Ford explorations in only $O(\beta)$ rounds. The fact that $\beta$-bounded distances can suffer $1+\epsilon$ stretch creates additional complications; one needs to define {\em approximate clusters}, and make sure that these approximate clusters correspond to actual trees in $G$. Finally, since the trees corresponding to $C(v)$ for the high level vertices $v\in A_i$, $i\ge k/2$, can have large depth, one needs to adapt the Thorup-Zwick routing scheme for trees \cite{TZ01-spaa}. In both \cite{EN16a,LPP16} this adaptation induced a logarithmic factor to both the table and the label size.

Our improved result has two main ingredients. First, we do not explicitly construct $G'$; In both \cite{EN16a,LPP16}, computing the weights of edges in $G'$ was a rather expensive step, and required large memory and induced a factor depending logarithmically on the aspect ratio to the running time. In addition, only approximate values were obtained. We observe that not all the edges of $G'$ are required for the algorithm, and thus we do not compute $G'$ at all. Rather we compute only those edges of $G'$ that are really needed for either the hopset or for the routing hierarchy. (This idea is reminiscent of \cite{E17}, where the virtual graph is also never entirely computed.)

Instead, we conduct the explorations in $G'$ by implementing in each iteration a $B$-bounded search in $G$, which not only saves memory and running time, but also simplifies the analysis, since now there is no error in the edge weights of $G'$. Second, our new tree-routing scheme has both improved label and routing table size, and can be computed with small memory.  (For more details, see \sectionref{sec:tree-route}.)  Our result is summarized below.

\begin{theorem}\label{thm:route}
Let $G=(V,E)$ be a weighted graph with $n$ vertices and hop-diameter $D$, and let $k>1$ be a parameter. Then there exists a routing scheme with stretch at most $4k-5+o(1)$, labels of size $O(k\log n)$ and routing tables of size $O(n^{1/k}\log n)$, that can be computed in a distributed manner within $(n^{1/2+1/k}+D)\cdot (\log n)^{O(k)}$ rounds, such that every vertex has memory of size $\tilde{O}(n^{1/k})$.

Alternatively, whenever $k\ge\sqrt{\log n/\log\log n}$, the number of rounds can be made $(n^{1/2+1/k}+D)\cdot 2^{\tilde{O}(\sqrt{\log n})}$ with memory $2^{\tilde{O}(\sqrt{\log n})}$ at each vertex.
\end{theorem}

In particular, taking $k=\delta\log n/\log\log n$ for a small constant $\delta$ yields $(n^{1/2+1/k}+D)\cdot n^{O(\delta)}$ rounds with $\polylog(n)$ memory per vertex.

\paragraph{Construction of Routing Scheme.}
Let $G=(V,E)$ be a weighted graph, fix $k> 1$. Sample a collection of sets $V=A_0\supseteq A_1\dots\supseteq A_k=\emptyset$, where for each $0<i<k$, each vertex in $A_{i-1}$ is chosen independently to be in $A_i$ with probability $n^{-1/k}$. A point $z\in A_i$ is called an $i$-pivot of $v$ if $d_G(v,z)=d_G(v,A_i)$. The cluster of a vertex $u\in A_i\setminus A_{i+1}$ is defined as
\begin{equation}\label{eq:cluster}
C(u)=\{v\in V~:~ d_G(u,v)<d_G(v,A_{i+1})\}~.
\end{equation}
It was shown in \cite{TZ01} that
\begin{claim}\label{claim:number-of-clusters}
With high probability, each vertex is contained in at most $4n^{1/k}\log n$ clusters.
\end{claim}

We recall a few definitions from \cite{EN16a}. For each $v\in V$ and $0\le i\le k-1$, a point $\hat{z}\in A_i$ is called an {\em approximate $i$-pivot} of $v$ if
\begin{equation}\label{eq:pivot}
d_G(v,\hat{z})\le (1+\epsilon)d_G(v,A_i)~.
\end{equation}
Define
\begin{equation}\label{eq:eps-clust}
C_\epsilon(u)=\{v\in V~:~ d_G(u,v)<\frac{d_G(v,A_{i+1})}{1+\epsilon}\}.
\end{equation}
The {\em approximate cluster} $\tilde{C}(u)$ will be any set that satisfies the following:
\begin{equation}\label{eq:app-cluster}
C_{6\epsilon}(u)\subseteq\tilde{C}(u)\subseteq C(u)~.
\end{equation}
It was shown in \cite{EN16a} that once we obtain approximate clusters as trees of $G$, with $\epsilon\le 1/(48k^4)$, and provide a routing scheme for these trees, it implies a routing scheme for $G$ with stretch $4k-5+o(1)$. In fact, it suffices that the routing scheme for each tree always routes through the root of the tree, not necessarily via the shortest path in the tree.

Let $h(u,v)$ denote the number of vertices on the shortest path from $u$ to $v$ in $G$. The following were also shown in \cite{EN16a} to hold with high probability.\footnote{For the sake of simplicity we will assume $k$ is even. For odd $k$, we can  improve the running time by a factor of $n^{1/(2k)}$.}
\begin{claim}\label{claim:hit-path}
For any $u,v\in V$ with $h(u,v)\ge B$, there exists a vertex of $A_{k/2}$ on the shortest path between them.
\end{claim}
\begin{claim}\label{cor:hit-path}
For any $0\le i< k-1$, $v\in A_i\setminus A_{i+1}$ and $u\in C(v)$, it holds that $h(u,v)\le 4n^{(i+1)/k}\ln n$.
\end{claim}
In particular, for $i< k/2$ we can find  the "exact" cluster $C(v)$ for each $v\in A_i\setminus A_{i+1}$, by a simple limited Bellman-Ford exploration from all such vertices $v$ to hop-depth  $4n^{(i+1)/k}\ln n\le\tilde{O}(\sqrt{n})$. By \claimref{claim:number-of-clusters}, the congestion induced at each $u\in V$ by the merit of being a part of many clusters is only $4n^{1/k}\log n$. So the total number of rounds required is $\tilde{O}(n^{1/2+1/k})$, and each vertex needs to store at most $4n^{1/k}\log n$ words (the clusters containing it). Finally, note that these clusters indeed correspond to trees, since every vertex $u\in C(v)$ can store as a parent the vertex who last updated the distance estimate that $u$ has for $v$.

From now on we consider the high levels, where $i\ge k/2$. Define $G'=(V',E')$ as a virtual graph where $V'=A_{k/2}$, and $E'$ corresponds to $B$-bounded distances in $G$. Observe that \claimref{claim:hit-path} implies that $d_{G'}(v,v')=d_G(v,v')$ for any $v,v'\in V'$ (because any shortest path in $G$ has a vertex of $V'$ within any $B$ hops on that path). First, we compute a $(\beta,\epsilon)$-hopset $H$ for the virtual graph $G'$ as in \theoremref{thm:dist-cong-path}, with parameters $\log k$, $\epsilon$ and $\rho=1/k$. If one desires the second assertion of the \theoremref{thm:route}, pick $\rho=\sqrt{\log\log n/\log n}$. Note that the graph $G'$ is implicit, and every node has internal memory $\tilde{O}(m^\rho)$. Since $|A_{k/2}|\le O(\sqrt{n})$ whp, the number of rounds required to compute $H$ is at most $(n^{1/2+1/k}+D)\cdot (\log n)^{O(1/\rho)}$ (recall $\rho\ge 1/k$ and $\epsilon\ge \Omega(1/\log^4n)$).

\paragraph{Approximate Pivots}
To compute the approximate pivots, conduct a Bellman-Ford exploration to depth $\beta$ in $G''=G'\cup H$, as in \lemmaref{lem:BF-small}, rooted in $A_{i+1}$, to compute for each $v\in V'$ a value $\hat{d}(v,A_{i+1})$. We perform another $B$-bounded exploration in $G$, where initially every vertex $v\in V'$ sends its current estimate, and in every step every vertex forwards the smallest value it has heard so far. We claim that every $u\in V$ will learn of an approximate ($i+1$)-pivot $\hat{z}\in A_{i+1}$. To see this, let $z$ be the ($i+1$)-pivot of $u$. If $h(u,z)\le B$, then $u$ will hear $z$'s message in the last $B$-bounded exploration. Otherwise, by \claimref{claim:hit-path}, there exists a vertex $v'\in V'$ on the shortest path from $u$ to $z$ within $B$ hops from $u$, and since $H$ is a $(\beta,\epsilon)$-hopset, we have that the first $\beta$ rounds of Bellman-Ford exploration from $A_{i+1}$ caused $v'$ to update $\hat{d}(v',A_{i+1})\le (1+\epsilon)d_G(v',A_{i+1})$. In the final exploration to range $B$, the vertex $v'$ will communicate this value on the path towards $u$. Thus, $u$ will have a value at most
\begin{equation}\label{eq:piv}
\hat{d}(u,A_{i+1})\le d_G^{(B)}(u,v')+\hat{d}(v',A_{i+1})\le d_G(u,v')+(1+\epsilon)d_G(v',A_{i+1})\le(1+\epsilon)d_G(u,A_{i+1})~,
\end{equation}
where the last inequality used that $d_G(u,v')+d_G(v',A_{i+1})=d_G(u,A_{i+1})$. This follows since $v'$ lies on the shortest path from $u$ to the nearest vertex of $A_{i+1}$.
We conclude that no matter which $\hat{z}$ is the approximate pivot of $u$, the distance estimate that $u$ has for it cannot be larger than $(1+\epsilon)d_G(u,A_{i+1})$.
Computing the approximate pivots requires $\tilde{O}(m^{1+\rho}+D)\cdot\beta=(n^{1/2+1/k}+D)\cdot (\log n)^{O(1/\rho)}$ rounds.

\paragraph{Approximate Clusters}

Fix some $i\ge k/2$, and for each $v\in A_i\setminus A_{i+1}$ we conduct a {\em limited} Bellman-Ford exploration in $G''=G'\cup H$ for $\beta$ rounds rooted at $v$, as in \lemmaref{lem:BF-small}. By ``limited", we mean that any vertex $u\in V'$ receiving a message originated at $v$, will forward it to its neighbors iff the current distance estimate is strictly less than $\hat{d}(u,A_{i+1})/(1+\epsilon)^2$.
We will refer to this condition, the {\em inclusion condition} of the exploration of $v$.
We need to avoid congestion at intermediate vertices during the $B$-bounded exploration in $G$ described in \lemmaref{lem:BF-small}, so these vertices will also need to implement some sort of limitation. Concretely, vertices $u\in V\setminus V'$ will forward $v$'s message iff their current estimate is strictly less than $\hat{d}(u,A_{i+1})/(1+\epsilon)$. The exploration over edges of $H$ is done as before, where \claimref{claim:number-of-clusters} guarantees every vertex participates in $4n^{1/k}\log n$ clusters (we will soon show that the approximate clusters are indeed contained in the clusters), so this bounds the number of rounds required by $\tilde{O}(n^{1/2+1/k}+D)\cdot\beta$. Also the memory per vertex required from this computation is bounded by $\tilde{O}(n^{1/k})$ (the number of cluster containing the vertex).

This exploration constructs a virtual tree rooted at $v$. For every edge $(x,y)\in E'$ on this tree, we add to the cluster all the vertices in $G$ on the $B$-bounded path from $x$ to $y$. This can be done via an acknowledgement message from $y$ back to $x$ on this path, and every vertex updates its parent accordingly.
For every hopset edge $(x,y)$ of the tree (which was broadcast to the entire graph during the exploration), every vertex $u\in P_{x,y}$, where $P_{x,y}$ is the path in $G$ implementing the edge $(x,y)$, joins the tree ($u$ knows about being a part of this edge by the path-reporting property of our hopset), and sets its distance estimate as $b_v(x)+d_P(x,u)$ if this value is smaller than its current estimate. If this is the case, the vertex $u$ also sets its parent as the neighbor on $P_{x,y}$ which is closer to $x$.

Finally, we perform another limited Bellman-Ford exploration to depth $B$ in $G$, where every vertex in the tree of $v$ sends its current distance estimate, and every vertex $u\in V$ will forward the smallest estimate it heard so far, but iff it is strictly less than $\hat{d}(u,A_{i+1})/(1+\epsilon)$. In that case it will also join the approximate cluster of $v$, and will update its parent as its neighbor in $G$ whose message caused $u$ to  update its distance estimate to $v$ for the last time.

Observe that the same vertex may join a tree more than once, due to several edges in $E'\cup H$ whose paths contain it. In such a case the vertex will have as a parent the vertex which minimize the estimated distance to the root. Since every vertex has a single parent, we will have that the approximate cluster of $v$, $\tilde{C}(v)$, is indeed a tree. It  remains to prove \eqref{eq:app-cluster}. Let $b_v(u)$ be the distance estimate that $u$ has to $v$ in the exploration rooted at $v$. 

\begin{claim}
For any $v\in V'$, $\tilde{C}(v)\subseteq C(v)$.
\end{claim}
\begin{proof}
Consider any $u\in\tilde{C}(v)$. 
If it is the case that $u\in V$ joined the approximate cluster by the exploration rooted at $v$, either by being in $V'$ or on a $B$-bounded path in $G$ that implements an edge of $E'$, then it must satisfy $b_v(u)<\hat{d}(u,A_{i+1})/(1+\epsilon)$. Now,
\[
d_G(u,v)\le b_v(u)<\hat{d}(u,A_{i+1})/(1+\epsilon)\stackrel{\eqref{eq:piv}}{\le} d_G(u,A_{i+1})~,
\]
so indeed $u \in C(v)$. The other case is that $u\ in P_{x,y}$ for a path $P_{x,y}$ implementing  a hopset edge $(x,y)$ that was added to the virtual tree. Since $y$ joins the approximate cluster, it must satisfy $b_v(y)<\hat{d}(y,A_{i+1})/(1+\epsilon)^2$.
Recall that the weight of the hopset edge $w_H(x,y)$ is the weight of the path $P = P_{x,y}$ from $x$ to $y$ in $G$ that $u$ lies on. Hence $d_P(x,u)+d_P(u,y)=w_H(x,y)$. It follows that
\begin{eqnarray*}
d_G(u,A_{i+1})&\stackrel{\eqref{eq:piv}}{\ge}&\frac{\hat{d}(u,A_{i+1})}{1+\epsilon}\ge \frac{d_G(u,A_{i+1})}{1+\epsilon}\ge \frac{d_G(y,A_{i+1})-d_G(u,y)}{1+\epsilon}\\
&\stackrel{\eqref{eq:piv}}{\ge}&\frac{\hat{d}(y,A_{i+1})}{(1+\epsilon)^2}-\frac{d_P(u,y)}{1+\epsilon}> b_v(y)-d_P(u,y)\\
&=& b_v(x)+w_H(x,y)-d_P(u,y) = b_v(x)+d_P(x,u)\\
&\ge& b_v(u)\ge d_G(u,v)~,
\end{eqnarray*}
where in the penultimate inequality we used the fact that the vertex $u$ knows $d_P(x,u)$, and thus it could have updated its distance estimate to $v$ as $b_v(x)+d_P(x,u)$ (note that it may have used a smaller estimate). Thus $u\in C(v)$ in this case, as required.
\end{proof}

The next claim proves the second inequality of (\ref{eq:app-cluster}).

\begin{claim}
For any $v\in V'$, $C_{6\epsilon}(v)\subseteq\tilde{C}(v)$.
\end{claim}
\begin{proof}
Let $u\in C_{6\epsilon}(v)$.  We would like to show that $u\in \tilde{C}(v)$. Consider the shortest path $P$ from $u$ to $v$ in $G$. Then by \claimref{claim:hit-path}, there is a vertex $u'\in V'$ on $P$ that is within $B$ hops from $u$. Notice that
\begin{equation}\label{eq:eue}
d_G(v,u')=d_G(v,u)-d_G(u,u')\le \frac{d_G(u,A_{i+1})-d_G(u,u')}{1+6\epsilon}\le\frac{d_G(u',A_{i+1})}{1+6\epsilon}~.
\end{equation}
Hence $u' \in C_{6\eps}(v)$ too.

We will show that the limited exploration originated at $v$ will reach $u'$, and in the final depth $B$ exploration it will reach $u$ and include it in $\tilde{C}(v)$.

Since $H$ is a $(\beta,\epsilon)$-hopset, there is a path $P'$ in $G''$ from $v$ to $u'$ that contains at most $\beta$ edges that satisfies
\begin{equation}\label{eq:uue}
d_{P'}(v,u')\le (1+\epsilon)d_{G'}(v,u')=(1+\epsilon)d_G(v,u')~.
\end{equation}
Let $z\in P'$ be any vertex on $P'$ that lies $t$ hops from $v$, $0 \le t \le \beta$.  Then after $t$ steps of Bellman-Ford exploration from $v$ we have that
\begin{eqnarray*}
b_z(v)&=&d_{P'}(v,z) = d_{P'}(v,u')-d_{P'}(z,u')\\
&\stackrel{\eqref{eq:uue}}{\le}&(1+\epsilon)d_G(v,u')-d_G(z,u')\stackrel{\eqref{eq:eue}}{\le}\frac{(1+\epsilon)d_G(u',A_{i+1})}{1+6\epsilon}-d_G(z,u')\\
&\le&\frac{d_G(u',A_{i+1})-d_G(z,u')}{1+4\epsilon}<\frac{d_G(z,A_{i+1})}{(1+\epsilon)^2}\le\frac{\hat{d}(z,A_{i+1})}{(1+\epsilon)^2}~.
\end{eqnarray*}
(We used that $\epsilon<1/5$.) We conclude that $z$ satisfies the inclusion condition for the exploration rooted at $v$, and forwards the message of $v$ onwards. In particular, by (\ref{eq:uue}), $b_v(u')\le d_{P'}(v,u')\le(1+\epsilon)d_G(v,u')$. In the final phase we make a Bellman-Ford exploration for $B$ rounds in $G$ from each vertex that received the  message of $v$. Thus, $u'$ will start such an exploration with distance estimate $b_v(u')$. Consider the subpath $Q\subseteq P$ from $u'$ to $u$. We have to show that every vertex on this path forwards the message of $v$, that is, that it satisfies the inclusion condition of the exploration of $v$. Let $y\in Q$ be such a vertex. Since this is a shortest path in $G$, we have
\begin{eqnarray*}
b_v(y)&\le& b_v(u')+d_Q(u',y)\le(1+\epsilon)d_G(v,u')+d_G(u',y)\\
&\le&(1+\epsilon)d_G(v,y) = (1+\epsilon)(d_G(v,u)-d_G(y,u))\\
&\stackrel{\eqref{eq:eps-clust}}{\le}&\frac{(1+\epsilon)d_G(u,A_{i+1})}{1+6\epsilon}-d_G(y,u)\le\frac{d_G(u,A_{i+1})-d_G(y,u)}{1+4\epsilon}\\
&\le&\frac{d_G(y,A_{i+1})}{1+4\epsilon}<\frac{\hat{d}(y,A_{i+1})}{1+\epsilon}~,
\end{eqnarray*}
as required.

\end{proof}

\subsection{Distributed Tree Routing with Small Memory}\label{sec:tree-route}

In this section we present our compact routing scheme for trees that can be computed in a distributed manner using small internal memory. In previous constructions of distributed routing schemes for trees \cite{EN16a,LPP16}, the internal memory was as high as $\sqrt{n}$, and it was also somewhat inefficient:  the label size is $O(\log^2n)$ and the routing tables are of size $O(\log n)$. Compare this to the classical \cite{TZ01-spaa} tree routing, which has label size $O(\log n)$ and routing tables of size $O(1)$. 

We follow the basic framework of previous works, by selecting a set $U\subseteq V$, such that each vertex is sampled to $U$ independently with probability $q$ ($q$ is a parameter, which we shall optimize later). Fix a tree $T$ on vertices $V(T)\subseteq V$ with root $z$. The vertices $U(T)=(U\cap V(T))\cup\{z\}$ partition the tree into subtrees, by removing the edges from each vertex in $U(T)$ to its parent. Each of the $|U(T)|$ subtrees is rooted in a vertex of $U(T)$. Denote by $T_w$ the subtree rooted at $w$. We also consider $T'$, the virtual tree on the vertices of $U(T)$, which is rooted at $z$, and contains an edge $(x,y)$ if the parent of $y$ lies in $T_x$. It is not hard to see (e.g., \cite{EN16a}) that whp the depth of each $T_w$ is $\tilde{O}(1/q)$,  and that $|U|\le O(qn)$.

In both \cite{EN16a,LPP16}, routing schemes were created for each $T_w$, and also a routing scheme for the virtual tree $T'$. This computation required large internal memory, since $z$ had to locally compute the scheme for $T'$. The inefficiency in the size was due to the fact that when routing in $T'$, traveling over a virtual edge $(x,y)$, one has to route in $T_x$ from $x$ to the parent of $y$. This seems to require storing additional routing information for this subtree, increasing both label and table size by a logarithmic factor. We overcome this issue by storing routing information only with respect to the actual tree, while applying pointer jumping techniques to quickly compute the full labels. However, we do not know how to construct exact tree routing with small memory. Fortunately, to implement our routing scheme for general graphs, it suffices to provide a {\em root-tree routing} scheme, where the routing is always  done via the root of the tree $T$, and not necessarily via the shortest path. (We stress that using larger memory, we can compute exact tree routing tables and labels within $\tilde{O}(\sqrt{n}+D)$ rounds, with label size $O(\log n)$ and routing tables of size $O(1)$, substantially improving previous results.)

Before describing our approach, let us briefly recall the Thorup-Zwick construction of tree routing. The idea is to assign to every (non-leaf) vertex $x\in T$ its {\em heavy child}, which is the child whose subtree has maximal size. Note that the subtree of any non-heavy child of $x$ contains at most half  of the vertices of  the  subtree  $T_x$  of $T$ rooted at $x$. For this reason, any path from the root $z$ to some $y\in T$ contains at most $\log n$ non-heavy edges. For an exact routing scheme they also conducts a DFS search in $T$ that assigns to each $y$ the DFS entry and exit times for its subtree.
The label of $y$ is these entry and exit times, and also the names of the non-heavy edges on the $z$ to $y$ path. The routing table $y$  consists of the DFS times, the name of the heavy child, and the name of the parent of $y$ in the tree. The routing towards a target $v$ in the tree is done as follows. At any intermediate vertex $y\in T$, if $v$ is not in the subtree rooted at $y$ (this can be checked via the DFS times), then $y$ forwards to its parent. If $v$ is in the subtree, $y$ inspects $v$'s label to see if an edge $(y,x)$ appears there.  If this is the case, it forwards to $x$, otherwise to its heavy child. Note that if one desires root-tree routing then there is no need to implement a DFS -- initially route to the parent until the root is reached, and then follow the path using heavy edges unless the label indicates otherwise.

Now we show how to implement our scheme in a distributed manner, and with $O(\log n)$ internal memory. 
First, every $w\in U(T)$ sends a message about itself to the vertices of $T_w$, informing them they are in $T_w$. Note that this message will arrive to all vertices in $T(U)$ who are children of $w$ in the virtual tree $T'$, so they will know their parent.
Next, for each $w\in U(T)$, every vertex in $T_w$ sends to its parent the size of the subtree rooted at it, beginning with the leaves. Every vertex that received messages from all its children, sums up the values and sends to its own parent. This can be done in parallel for all trees $T_w$ for $w\in U(T)$, and will take $\tilde{O}(1/q)$ (the bound on the height of each $T_w$) rounds.

For a vertex $v$ in a tree $T$, rooted at a vertex $z$, and a positive integer $h$, we say that a vertex $u$ is an {\em $h$-ancestor} of $v$, if $u$ lies on the unique $v-z$ path in $T$ at distance $h$ from $v$.

We would like that every $y\in T$ will know the entire size of the subtree of $T$ rooted at $y$. Initially, we compute this value only for the virtual vertices of $U(T)$. For  a vertex  $x \in U(T)$, its subtree size is exactly the sum of sizes of subtrees $T_w$ for $w$ that are in the subtree of $T'$ rooted at $x$. Note that computing these values from the leaves of $T'$ up will not be efficient, since every message on a virtual edge may require $O(D)$ rounds, and the depth of $T'$ may be as large as $qn$ (which will be approximately $\sqrt{n}$). Thus, this  results in $O(D\sqrt{n})$ rounds. To alleviate this issue, we use the following "pointer jumping" technique. Initially, set for $x\in U(T)$ the current size $s_x=|T_x|$, and its first ancestor $a_1(x)$ as its parent in $T'$ (and for the root $z$, set $a_1(z)=\bot$). For $i=0,1,\dots,\log n$ rounds, every vertex $x\in U(T)$ will broadcast in the $i$th round (using the BFS tree of $G$), the current size $s_x$ and the name of its $2^i$-ancestor $a_i(x)$ in $T'$. Then whenever $x$ hears a message that some $w\in U(T)$ broadcasts with $x=a_i(w)$, then $x$ adds $s_w$ to its current size $s_x$. In addition,
the vertex $x$ hears the message of $a_i(x)$, and it
 updates $a_{i+1}(x)$ as $a_i(a_i(x))$. (It could be the case that $a_i(a_i(x))=\bot$.  In this case, indeed, $a_{i+1}(x)=\bot$.) We claim that this process correctly computes for any $x\in U(T)$ the size of the subtree of $T$ rooted at $x$. It can be shown by induction on $i$, that before the $i$th round, $s_x$ is the size of the subtree rooted at $x$ that contains at most $2^i$ vertices of $U(T)$ on any root-leaf path.
There are $O(|U(T)|)\le\tilde{O}(qn)$ messages sent on each round for $\log n$ rounds. Hence, it will take $\tilde{O}(qn+D)$ rounds to implement this step.

In order to compute $s_y$, the size of the subtree of $T$ rooted at $y$, for all $y\in T$, every $x\in U(T)$ informs its parent in $T$ with the value $s_x$. Then once again, for every $w\in U(T)$ in parallel, the leaves of $T_w$ start to send to their parent their current size. This time, some of these leaves and internal vertices could be parents of vertices in $U(T)$, so these sizes are the actual subtree size in $T$. In $\tilde{O}(1/q)$ rounds, every vertex $y\in T$ will know $s_y$. After sending these values to the parents, every vertex can infer who is its heavy child.

The label $L(y)$ needed for root-tree routing is just the collection of edges $\{(u,v)\}$ that are on the $z-y$ path in $T$, such that $v$ is not the heavy child of $u$. Clearly, there can be at most $\log n$ such edges on this path, because the size of the subtree decreases by a factor of 2 for every non-heavy edge. If $y\in T_x$, we start by computing a partial label that contains non-heavy edges on the path from $x$ to $y$. This can be done by initializing $L(x)=\emptyset$, and starting at $x$, any vertex $u\in T_x$ which received a label $L(u)$, sends $L(u)$ to its heavy child, and $L(u)\cup\{(u,v)\}$ for any non-heavy child $v$. These labels are also sent to the children of $x$ in $T'$ (recall that these are the vertices $T(U)$ whose $T$-parents belong to $T_x$). Once this computation is completed, every vertex $w\in T(U)$ knows the non-heavy edges on the path from $x$, its parent in $T'$, to $w$. We again apply pointer jumping to compute the full labels. For $i=0,1,\dots,\log n$, every vertex of $U(T)$ will broadcast in the $i$th round its current label. In each round, when $x$ hears the message from its $2^j$-ancestor $a_j(x)$ (recall that $x$ computed previously its $2^j$-ancestors, for all $j = 0,1,\ldots,\log n$, and it stored them in its internal memory), it will update $L(x)\leftarrow L(a_j(x))\cup L(x)$. Once again, it can be proved by induction on $i$ that before the $i$th round, every $x\in U(T)$ knows all the non-heavy edges on the path in $T$ from $a_i(x)$ to $x$ (or from the root $z$ to $x$ if $a_i(x)=\bot$). Since every label has size $O(\log n)$, this will require $\tilde{O}(qn+D)$ rounds.
Finally, in another $\tilde{O}(1/q)$ rounds, each $x\in U(T)$ sends its updated label $L(x)$ to every vertex $y\in T_x$, and they update their label by appending $L(x)$.

If one desires a routing scheme for a single tree, just take $q=1/\sqrt{n}$, so the running time will be $\tilde{O}(\sqrt{n}+D)$. If we desire to compute a routing scheme in parallel for multiple trees, but have the guarantee that every $v\in V$ belongs to at most $s$ trees, then we can use the argument as in \cite{EN16} to obtain running time $\tilde{O}(\sqrt{s\cdot n}+D)$ (rather than the naive $\tilde{O}(s\cdot\sqrt{n}+D)$). We conclude by formally summarizing our result.

\begin{theorem}\label{thm:tree-route}
For any tree $T$ on $n$ vertices, lying in a network with hop-diameter $D$, there exists a distributed algorithm in the CONGEST model running in $\tilde{O}(\sqrt{n}+D)$ rounds, that computes a root-tree routing scheme with label size $O(\log n)$ and routing tables of size $O(1)$, such that every vertex uses only $O(\log n)$ words of memory throughout the computation.

Moreover, if there are no restriction on the memory used throughout the computation, then exact tree routing tables of size $O(1)$ and labels of size $O(\log n)$ can be computed in $\tO(\sqrt{n} + D)$ time.

In addition, given a network with $n$ vertices and a set of trees so that each vertex is contained in at most $s$ trees, one can compute a root-tree routing scheme as above for all trees in parallel, within $\tilde{O}(\sqrt{s\cdot n}+D)$ rounds, while using memory $O(s\cdot\log n)$ at each vertex.

\end{theorem}

\section*{Acknowledgements}

We wish to thank Christoph Lenzen for raising to us the problem of distributed routing with small individual memory requirements, and for permitting us to use a quotation from \cite{L16}.

\bibliographystyle{alpha}
\bibliography{hopset}

\end{document}